\documentclass[a4paper,UKenglish, autoref, cleveref, thm-restate]{lipics-v2021}
\usepackage{bbm}
\usepackage{multirow, siunitx}
\nolinenumbers

\newcommand{\indep}{\perp \!\!\! \perp}

\def\equationautorefname~#1\null{%
	(#1)\null
}
\usepackage{graphicx}
\usepackage{color}
\usepackage{subcaption}
\usepackage{graphics}

\usepackage{enumerate}

\usepackage{amssymb}
\usepackage{amsmath}
\usepackage{amsfonts}
\usepackage{bm}
\usepackage{mathtools}
\usepackage{cuted}
\usepackage{verbatim}

\usepackage{textcomp}

\usepackage{multirow}

\usepackage{booktabs}

\usepackage{url}

\usepackage{balance}

\usepackage[lined,boxed,linesnumbered,vlined]{algorithm2e}
\usepackage{algpseudocode}

\usepackage{etoolbox}

\SetKwInput{Problem}{problem}
\SetKwInput{Input}{input}
\SetKwInput{Output}{output}
\SetAlgoSkip{smallskip} 
\SetAlgoInsideSkip{smallskip} %
\SetKwComment{tcp}{\#\ }{}
\SetKwComment{tcc}{\#\ }{}
    
\usepackage{amsthm}
\newtheoremstyle{saber}
  {}
  {}
  {\itshape}
  {}
  {\bfseries}
  {.}
  {.5em}
  {}
\theoremstyle{saber}
\ifcsundef{theorem}{
\newtheorem{theorem}{Theorem}[section]
}{}
\ifcsundef{corollary}{%

}{}
\ifcsundef{lemma}{%
\newtheorem{lemma}{Lemma}[section]
}{}
\ifcsundef{definition}{%
\newtheorem{definition}{Definition}[section]
}{}
\ifcsundef{fact}{%

}{}
\ifcsundef{remark}{%
\newtheorem{remark}{Remark}
}{}
\ifcsundef{property}{%
}{}

\ifx\theoremautorefname\undefined
\newcommand*{\theoremautorefname}{Theorem}
\fi
\ifx\lemmaautorefname\undefined
\newcommand*{\lemmaautorefname}{Lemma}
\fi
\ifx\definitionautorefname\undefined
\newcommand*{\definitionautorefname}{Definition}
\fi
\ifx\corollaryautorefname\undefined
\newcommand*{\corollaryautorefname}{Corollary}
\fi
\ifx\factautorefname\undefined
\newcommand*{\factautorefname}{Fact}
\fi
\ifx\propertyautorefname\undefined
\newcommand*{\propertyautorefname}{Property}
\fi

\ifx\argmax\undefined
\newcommand{\argmax}{\operatornamewithlimits{argmax}}
\fi
\ifx\argmin\undefined
\newcommand{\argmin}{\operatornamewithlimits{argmin}}
\fi
\ifx\limsup\undefined
\newcommand{\limsup}{\operatornamewithlimits{limsup}}
\fi
\ifx\liminf\undefined
\newcommand{\liminf}{\operatornamewithlimits{liminf}}
\fi
\ifx\norm\undefined
\newcommand{\norm}[1]{\lVert#1\rVert}
\fi
\ifx\abs\undefined
\newcommand{\abs}[1]{\lvert#1\rvert}
\fi
\ifx\set\undefined
\newcommand{\set}[1]{\left\{#1\right\}}
\fi
\ifx\mset\undefined
\newcommand{\mset}[1]{\lbrack #1\rbrack}
\fi
\ifx\etal\undefined
\newcommand{\etal}[1]{{\em #1 et al.}~}
\fi
\ifx\ie\undefined
\newcommand{\ie}{{\em i.e.,} }
\fi
\ifx\eg\undefined
\newcommand{\eg}{{\em e.g.,} }
\fi
\ifx\etc\undefined
\newcommand{\etc}{{\em etc.,} }
\fi
\ifx\wrt\undefined
\newcommand{\wrt}{{\em w.r.t.} }
\fi
\ifx\dotprod\undefined
\newcommand{\dotprod}[2]{
  \langle #1, #2 \rangle
}
\fi
\ifx\iid\undefined 
\newcommand{\iid}{i.i.d.}
\fi
\ifx\bigParenthes\undefined 
\newcommand{\bigParenthes}[1]{
  \big(#1\big)
}
\fi
\ifx\bigBracket\undefined 
\newcommand{\bigBracket}[1]{
  \big\{#1\big\}
}
\fi
\ifx\bigSqBracket\undefined 
\newcommand{\bigSqBracket}[1]{
  \big[#1\big]
}
\fi
\ifx\BigParenthes\undefined 
\newcommand{\BigParenthes}[1]{
  \Big(#1\Big)
}
\fi
\ifx\BigBracket\undefined 
\newcommand{\BigBracket}[1]{
  \Big\{#1\Big\}
}
\fi
\ifx\BigSqBracket\undefined 
\newcommand{\BigSqBracket}[1]{
  \Big[#1\Big]
}
\fi
\ifx\biggParenthes\undefined 
\newcommand{\biggParenthes}[1]{
  \bigg(#1\bigg)
}
\fi
\ifx\biggBracket\undefined 
\newcommand{\biggBracket}[1]{
  \bigg\{#1\bigg\}
}
\fi
\ifx\biggSqBracket\undefined 
\newcommand{\biggSqBracket}[1]{
  \bigg[#1\bigg]
}
\fi
\ifx\BiggParenthes\undefined 
\newcommand{\BiggParenthes}[1]{
  \Bigg(#1\Bigg)
}
\fi
\ifx\BiggBracket\undefined 
\newcommand{\BiggBracket}[1]{
  \Bigg\{#1\Bigg\}
}
\fi
\ifx\BiggSqBracket\undefined 
\newcommand{\BiggSqBracket}[1]{
  \Bigg[#1\Bigg]
}
\fi
\ifx\bracket\undefined
\newcommand{\bracket}[1]{
  \{#1\}
}
\fi
\ifx\parenthes\undefined
\newcommand{\parenthes}[1]{
  (#1)
}
\fi
\ifx\sqBracket\undefined
\newcommand{\sqBracket}[1]{
  [#1]
}
\fi
\ifx\prob\undefined 
\newcommand{\prob}[1]{\mathbb{P}[#1]}
\fi
\ifx\Prob\undefined 
\newcommand{\Prob}[1]{\mathbb{P}\big[#1\big]}
\fi
\ifx\probb\undefined 

\fi
\ifx\expect\undefined 
\newcommand{\expect}[1]{\mathbb{E}[#1]}
\fi
\ifx\Expect\undefined 
\newcommand{\Expect}[1]{\mathbb{E}\big[#1\big]}
\fi
\ifx\expectt\undefined 
\newcommand{\expectt}[1]{\mathbb{E}\bigg[#1\bigg]}
\fi
\ifx\walk\undefined 
\makeatletter
\newcommand{\walk}[1]{%
  \@tempswafalse
  \@for\next:=#1\do
    {\if@tempswa\!\!\rightarrow\!\!\else\@tempswatrue\fi\next}%
}
\makeatother
\fi
\ifx\seq\undefined 
\newcommand{\seq}{\!=\!}
\fi
\ifx\sminus\undefined 
\newcommand{\sminus}{\!-\!}
\fi
\ifx\sm\undefined 
\newcommand{\sm}[1]{\!#1\!}
\fi

\ifx\union\undefined
\newcommand{\union}[2]{#1\!\cup\!#2}
\fi

\ifx\hl\undefined 
\newcommand{\hl}[2]{{\color{#1}#2}}
\fi

\ifx\hlb\undefined 
\newcommand{\hlb}[1]{{\color{blue}#1}}
\fi

\ifx\hlr\undefined 
\newcommand{\hlr}[1]{{\color{red}#1}}
\fi

\ifx\hlp\undefined 
\newcommand{\hlp}[1]{{\color{purple}#1}}
\fi

\ifx\hlc\undefined 
\newcommand{\hlc}[3]{{\color{#1}#2 [remark: #3]}}
\fi

\ifcsdef{sitemize}{}{

}

\ifcsdef{senumerate}{}{

}

\ifcsdef{spenumerate}{}{

}

\ifcsdef{slenumerate}{}{

}

\title{On Efficient Range-Summability of IID Random Variables in Two or Higher Dimensions}
\titlerunning{On Efficient Range-Summability of IID RVs in Multiple Dimensions}
\author{Jingfan Meng}{School of Computer Science, Georgia Institute of Technology, Atlanta, USA}{jmeng40@gatech.edu}{}{}
\author{Huayi Wang}{School of Computer Science, Georgia Institute of Technology, Atlanta, USA}{huayiwang@gatech.edu}{}{}
\author{Jun Xu}{School of Computer Science, Georgia Institute of Technology, Atlanta,  USA}{jx@cc.gatech.edu}{}{}
\author{Mitsunori Ogihara}{Department of Computer Science, University of Miami, Coral Gables,  USA}{ogihara@cs.miami.edu}{}{}

\begin{document}
	

\authorrunning{J. Meng and H. Wang and J. Xu and M. Ogihara} 

\Copyright{Jingfan Meng, Huayi Wang, Jun Xu, and Mitsunori Ogihara} 



\relatedversiondetails{\normalfont A previous version of this paper is available at} {https://arxiv.org/abs/2110.07753}
\funding{This material is based upon work supported by the National Science Foundation under Grant No. CNS-1909048, CNS-2007006, CNS-2051800, and by Keysight Technologies under Grant No. BG005054.}
\acknowledgements{}

\maketitle
\begin{abstract}
$d$-dimensional (for $d>1$) efficient range-summability ($d$D-ERS) of random variables (RVs) is a fundamental algorithmic problem that has applications to two important families of database problems, namely,
fast approximate wavelet tracking (FAWT) on data streams and approximately answering range-sum queries over a data cube. 
Whether there are efficient solutions to the $d$D-ERS problem, or to the latter database problem, have been two long-standing open problems.  Both are solved in this work.
Specifically, we propose a novel solution framework to $d$D-ERS on RVs that have Gaussian or Poisson distribution.  Our $d$D-ERS solutions are the first ones that have polylogarithmic time complexities. 
 Furthermore, we develop a novel $k$-wise independence theory that allows our $d$D-ERS solutions to have both high computational efficiencies and strong provable independence guarantees. Finally, 
 we show that under a sufficient and likely necessary condition, certain existing solutions for 1D-ERS can be generalized to higher dimensions.

\end{abstract}


\section{Introduction}\label{sec:intro}

Efficient range-summability (ERS) of  random variables (RVs) is a fundamental algorithmic problem that has been studied for nearly two 
decades~\cite{feigenbaum, rusu-jour, dyahist, 1ddyadic}.  This problem has so far been defined only in one dimension (1D) as follows.
Let $X_0, X_1, \cdots, X_{\Delta-1}$ be a list of 
 \emph{underlying RVs} each of which has the same {\it target distribution} $X$.
Here, the (index) universe size $\Delta$ is typically a large number (say $\Delta = 2^{64}$).  
A 1D-ERS problem calls for the following oracle for answering range-sum queries over (realizations of) these underlying RVs.  
At initialization, the oracle chooses a {\it random outcome} $\omega$ from the {\it sample space} $\Omega$, which
{\it mathematically determines} the (values of the) realizations $X_0(\omega), X_1(\omega), \cdots, X_{\Delta-1}(\omega)$;   here the phrase 
“mathematically determines” emphasizes that (an implementation of) the oracle does not actually realize these RVs (and pay the $O(\Delta)$ time cost) at initialization.   
Thereafter, given any query range $[l, u) \triangleq \{l, l+1, \cdots, u-1\}$ that lies in the universe $[0, \Delta)$, the oracle is required to return 
$S[l, u) \triangleq \sum_{i=l}^{u-1}X_i(\omega)$, the sum of the realizations of all underlying RVs in the range.   This requirement is called the \emph{consistency}
requirement, which is one of the two essential requirements for the ERS oracle.   We will show that such an ERS oracle can be efficiently implemented using 
hash functions.  With such an implementation, the outcome $\omega$ corresponds to the seeds of these hash functions. 


The other essential requirement is \emph{correct distribution}, which has two aspects.   
The first aspect is that the underlying RVs $X_0, X_1, \cdots, X_{\Delta-1}$ each has the same target (marginal) distribution $X$. 
The second aspect is that these RVs should satisfy certain independence guarantees. 
Ideally, it is desired for these RVs to be mutually independent, but this comes at a high storage cost as we will elaborate shortly.
In practice, another type of independence guarantee, namely $k$-wise independence (in the sense that any subset of $k$ underlying RVs are independent),
is good enough for most applications when $k \ge 4$.   We will show that our solution for ERS in $d>1$ dimensions 
can provide $k$-wise independence guarantee at a small storage cost of $O(\log^d\Delta)$ for an arbitrarily large $k$.  
 

A straightforward but naive way to answer a range-sum query, say over $[l, u)$, is simply to sum up the realization of every underlying RV
$X_l(\omega), X_{l+1}(\omega), \cdots,$ $X_{u-1}(\omega)$ in the query range.
This solution, however, has a time complexity of
$O(\Delta)$ when $u - l$ is $O(\Delta)$.
In contrast, an efficient solution should be able to do so with only $O(\mathrm{polylog}(\Delta))$ time complexity.
Indeed, all existing ERS solutions~\cite{calderbank, feigenbaum, rusu-jour, dyahist, 1ddyadic} have $O(\log\Delta)$ time complexity.

\subsection{Related Work on 1D-ERS}\label{sec:1ders}



There are in general two families of solutions to the ERS problem in 1D, following two different approaches.  
The first approach is based on error correction codes (ECC).
Solutions taking this approach include BCH3~\cite{rusu-jour}, EH3~\cite{feigenbaum}, and RM7~\cite{calderbank}. 
This approach has two drawbacks.  First, it works only when the target distribution $X$ is Rademacher.
Second, although it guarantees $3$-wise (in the case of BCH3 and EH3) or $7$-wise (in the case of RM7) independence among the underlying RVs, 
almost all empirical independence beyond that is destroyed. 
In addition, RM7 is very slow in practice~\cite{rusu-jour}.


The second approach is based on a data structure called dyadic simulation tree (DST), which we will 
describe in~\autoref{sec:dst1D}.  The DST-based approach was first briefly mentioned in~\cite{dyahist} and later fully developed in~\cite{1ddyadic}.  
The DST-based approach is better than the ECC-based approach in two aspects.  First, it
supports a wider range of target distributions including Gaussian, Cauchy, Rademacher~\cite{1ddyadic}, and Poisson (see~\autoref{sec:2ddyasim}). 
Second, it provides stronger independence guarantees at a low computational cost.  
For example, when implemented using the tabulation hashing scheme~\cite{tabulation5wise}, it guarantees $5$-wise independence at a much lower computational
cost than RM7~\cite{1ddyadic}.  We will describe a nontrivial generalization of this result to $2$D in~\autoref{sec:kwise}.




\subsection{ERS in Higher Dimensions}

In this work, we formulate the ERS problems in $d>1$ dimensions ($d$D), which we denote as $d$D-ERS, and propose the first-ever solutions to $d$D-ERS.
A $d$D-ERS problem is similarly defined on a $d$-dimensional universe $[0, \Delta)^d$ that contains $\Delta^d$ integral points.  
Each $d$D point $\vec{\mathbf{i}}  \in [0,\Delta)^d$ is associated with
an RV $X_{\vec{\mathbf{i}}}$, and 
every such RV has the same target (marginal) distribution $X$. 
Here, for ease of presentation, we assume $\Delta$ is the same on each dimension and is a power of $2$, but our solutions can work without these two assumptions.
Let $\vec{\mathbf{l}}=(l_1, l_2, \cdots, l_d)^T$ and $\vec{\mathbf{u}} = (u_1, u_2, \cdots, u_d)^T$ be two $d$D points in $[0, \Delta)^d$ such that 
$l_j < u_j$ for each dimension $j=1,2,\cdots, d$.  We define $[\vec{\mathbf{l}}, \vec{\mathbf{u}})$ as the $d$D rectangular range ``cornered'' by these two points 
in the sense $[\vec{\mathbf{l}}, \vec{\mathbf{u}}) \triangleq [l_1, u_1)\times [l_2, u_2)\times \cdots\times [l_d, u_d)$, where $\times$ is the Cartesian product. 

A $d$D-ERS problem calls for the following oracle.  At initialization, the oracle chooses an outcome $\omega$ 
that \emph{mathematically determines} the realization $X_{\vec{\mathbf{i}}}(\omega)$ for each $\vec{\mathbf{i}}\in [0, \Delta)^d$. 
Thereafter, given any $d$D range $[\vec{\mathbf{l}}, \vec{\mathbf{u}})$, the oracle needs to return 
in $O(\mathrm{polylog}(\Delta))$ time $S[\vec{\mathbf{l}}, \vec{\mathbf{u}}) \triangleq \sum_{\vec{\mathbf{i}}\in [\vec{\mathbf{l}}, \vec{\mathbf{u}})} X_{\vec{\mathbf{i}}}(\omega)$, the sum of the realizations of all underlying RVs in this $d$D range.
Unless otherwise stated, the vectors that appear in the sequel are assumed to be column vectors. 
We write them in boldface and with a rightward arrow on the top like in ``$\vec{\mathbf{x}}$''.

Several 1D-ERS solutions have been proposed as an essential building block for efficient solutions to several database problems.  
In two such database problems that we will describe in~\autoref{sec:introapplication}, their 1D solutions, both proposed in~\cite{onepasswavelet}, 
can be readily generalized to $d$D if their underlying 1D-ERS oracles can be generalized to $d$D.  
In fact, in~\cite{muthukrishnanhistogram}, 
authors stated explicitly that the only missing component for their solutions of the 1D database problems to be generalized to 2D was an efficient 
2D-ERS oracle where $X$ is the Rademacher distribution ($\Pr[X = 1] = \Pr[X = 0] = 0.5$, aka. single-step random walk).
However, until this paper, \emph{no solution} to any $d$D-ERS problem for $d >1$ 
has been proposed.

\subsection{Our \boldmath{d}D-ERS Solutions}\label{sec:waveletapproach}

In this paper, we propose novel solutions to the two $d$D-ERS problems wherein the target distributions are Gaussian and Poisson respectively. 
We refer to these two problems as $d$D Gaussian-ERS and $d$D Poisson-ERS, respectively. 
Both solutions generalize the corresponding DST-based 1D-ERS solutions to higher dimensions and 
have a low time complexity of $O(\log^d\Delta)$ per range-sum query. 
Our $d$D Gaussian-ERS solution, in particular, is based on the Haar wavelet transform (HWT), since DST is equivalent to HWT when (and only when) the target distribution $X$ is Gaussian,
as will be shown in~\autoref{sec:wavelet1d}. 

Furthermore, we identify a sufficient condition that, if satisfied by the target distribution $X$, 
guarantees that the corresponding DST-based 1D-ERS solution can be generalized to a $d$D-ERS solution.   
We will prove in~\autoref{sec:casepositive} that Gaussian and Poisson are two ``nice'' distributions that satisfy this sufficient condition.
We will also show that, for all such ``nice'' distributions  (including those
we might discover
in the future),
this generalization process (from 1D to $d$D) follows a universal algorithmic framework that can be characterized as the Cartesian product of $d$ DSTs.   
We will also provide
strong evidence that $X$ ``being nice'' is likely necessary for this DST generalization (from 1D to $d$D) to be feasible (see~\autoref{sec:infeasibility}).



Unfortunately, so far we have not found any ``nice'' distribution other than Gaussian and Poisson.
Hence $d$D-ERS for other target distributions remains an open problem, and is likely not solvable by the (generalized) DST approach. 
We emphasize this is not a shortcoming of the DST approach:  That we have obtained computationally efficient solutions in the cases of Gaussian and Poisson 
is already a pleasant surprise, as the $d$D-ERS problem has been open for nearly two decades.    
Furthermore, 
we will show that our $d$D Gaussian-ERS solution leads to computationally efficient solutions to both aforementioned database problems (to be described in~\autoref{sec:introapplication}),
by answering their calls for a $d$D Gaussian-ERS or equivalent oracle.

Our $d$D Gaussian-ERS and Poisson-ERS solutions both support two different types of
independence guarantees, at different storage costs.   
The first type is the ideal case in which the $\Delta^d$ underlying RVs are mutually independent. 
As will be shown in~\autoref{sec:dyasim}, we can achieve this ideal case by paying $O(T\log^d \Delta)$ storage cost, where 
$T$ is the total number of range-sum queries to be answered (i.e., $O(\log^d \Delta)$ storage cost per range query). 
The second type is also quite strong:  The $\Delta^d$ underlying RVs are $k$-wise independent, where the constant $k$ can be arbitrarily large.
In~\autoref{sec:kwise}, we propose a $k$-wise independence scheme that can provide the second type of guarantees by 
employing $O(\log^{d}\Delta)$ $k$-wise 
independent hash functions.
Its storage cost is quite small:  only $O(\log^{d}\Delta)$ for storing the seeds of these hash functions.
We emphasize that the issue of how strong this independence guarantee (among the underlying RVs) needs to be 
affects only the storage cost of our Gaussian-ERS and Poisson-ERS solutions, and is orthogonal to all other issues described
in earlier paragraphs such as the $O(\log^d\Delta)$ time complexity of both solutions and 
the sufficient and likely necessary condition for a DST-based $d$D-ERS solution to exist. 

This $k$-wise independence scheme makes our $d$D-ERS solutions very practically useful for two reasons. 
First, such a $k$-wise independent hash function in practice requires a very short seed (not longer than a few kilobytes), and each hash operation can be computed in 
nanoseconds~\cite{univhash, tabulationpower}.
Second, most applications of ERS only require the underlying RVs to be 4-wise independent~\cite{onepasswavelet, muthukrishnanhistogram}.

The contributions of this work can be summarized as follows.  First, we provide the first set of answers to the long-standing open question whether there is an efficient 
solution to {\it any} $d$D-ERS problem for $d > 1$.   Second, our Gaussian-ERS solution solves a long-standing open problem in data streaming
that we will describe next.
Third, our $k$-wise independence theory and hashing scheme make our $d$D ERS solutions 
very practically useful.  


The rest of the paper is organized as follows. In~\autoref{sec:introapplication}, we describe two applications of our $d$D Gaussian-ERS solutions.  In~\autoref{sec:dyasim}, we first describe our HWT-based Gaussian-ERS scheme in $1$D, and then generalize it to 2D and $d$D. In~\autoref{sec:kwise}, we describe our $k$-wise independence theory and scheme. In~\autoref{sec:infeasibility}, we propose a sufficient and likely necessary condition on the target distribution for the DST approach to be generalized to $d$D. Finally, we conclude the paper in~\autoref{sec:conclusion}.

\section{Applications of $\bf{d}$D Gaussian-ERS}\label{sec:introapplication}

In this section, we introduce two important applications of our $d$D Gaussian-ERS solution. 

\subsection{Fast Approximate Wavelet Tracking}\label{sec:FAWT}

The first application is to the problem of fast approximate wavelet tracking (FAWT) on data streams~\cite{onepasswavelet, cormodewavelettracking}. 
We first introduce the FAWT problem in 1D~\cite{onepasswavelet}, or 1D-FAWT for short.  
In this problem, the precise system state is comprised of a $\Delta$-dimensional vector $\vec{\mathbf{s}}$, each scalar of which is a counter.  
The precise system state at any moment of time is determined by a data stream, in which each data item is an update to 
one such counter (called a {\it point update}) or all counters in a 1D range (called a {\it range update}).  In 1D-FAWT, $\vec{\mathbf{s}}$ is considered
a $\Delta$-dimensional signal vector that is constantly ``on the move'' caused by the updates in the data stream.  
Let $\vec{\mathbf{r}}$ be the ($\Delta$-dimensional) vector of HWT coefficients of $\vec{\mathbf{s}}$. 
Clearly, $\vec{\mathbf{r}}$ is also a ``moving target''.  We denote as $\vec{\mathbf{r}}_t$ the snapshot of $\vec{\mathbf{r}}$ at a time $t$.
In 1D-FAWT, the goal is to closely track (the precise value of) $\vec{\mathbf{r}}$ over time using a {\it sketch}, in the sense that at moment $t$, 
we can recover from the sketch
an estimate $\vec{\mathbf{r}}'_t$ of $\vec{\mathbf{r}}_t$, such that $\|\vec{\mathbf{r}}_t - \vec{\mathbf{r}}'_t\|_2$ is small.   An acceptable solution should use a sketch whose size (space complexity)  is
only $O(\mathrm{polylog}(\Delta))$, and be able to maintain the sketch with a computation time cost of $O(\mathrm{polylog}(\Delta))$ per point or range update.  

The first solution to 1D-FAWT was proposed in~\cite{onepasswavelet}.  It requires the efficient computation of an arbitrary scalar in $H \vec{\mathbf{x}}$, where $H$ is the $\Delta\times\Delta$ Haar matrix (to be defined in~\autoref{sec:1dmath}) and 
$\vec{\mathbf{x}}$ is a $\Delta$-dimensional vector of 4-wise independent Rademacher RVs.
A key step of this computation is 
to compute a range-sum of 4-wise independent Rademacher RVs (in $1$D), that is used therein as a Tug-of-War (ToW) sketch~\cite{ams} for ``sketching'' the $L_2$ difference (approximation error) between 
the signal vector and its FAWT approximation.
An aforementioned ECC-based ERS solution is used therein to tackle this Rademacher-ERS problem. 
Authors of~\cite{muthukrishnanhistogram} stated that if they could find a solution to this Rademacher-ERS problem in $d$D, then the 1D-FAWT solution in~\cite{onepasswavelet} would become a
$d$D-FAWT solution.  
The first solution to $d$D-FAWT, proposed in~\cite{cormodewavelettracking}, explicitly bypassed this ERS problem.

We note that the 1D-FAWT solution above continues to work, and its time and space complexities remain the same, 
if we replace the $\vec{\mathbf{x}}$ with a $\Delta$-dimensional vector of 
4-wise independent standard Gaussian RVs.  This is because, with this replacement, the aforementioned ToW sketch becomes a Gaussian Tug-of-War (GToW) sketch (which maps a data item to a 
Gaussian RV instead of a Rademacher RV)~\cite{stabledist},
and ToW and GToW are known to have the same $(\epsilon, \delta)$ accuracy bound~\cite{ams, stabledist} for sketching the $L_2$ norm of a data stream (used here for sketching the 
aforementioned $L_2$ difference). 
Based on this insight, our $d$D Gaussian-ERS solution can be used to construct a $d$D-FAWT solution as follows.  
We simply change, in the {\it contingent} $d$D-FAWT solution proposed in~\cite{onepasswavelet}, the distribution of all $\Delta^d$ underlying 4-wise independent RVs 
from Rademacher to Gaussian.  
With this replacement, this {\it contingent} solution will finally work, {\it provided} we can solve the resulting $d$D Gaussian-ERS problem.
The latter problem is solved by our $k$-wise (with $k=4$ here) independence scheme, to be described in~\autoref{sec:kwise}.  
The resulting $d$D-FAWT solution has the same time and space complexity of $O(\log^d\Delta)$
as that proposed in~\cite{cormodewavelettracking} for achieving the same accuracy guarantee.   


	

\subsection{Range-Sum Queries over Data Cube}\label{sec:datacube}
Our second application is to the problem of approximately answering range-sum queries over a data cube~\cite{datacube} that is similarly ``on the move'' propelled by the (point or range) updates that arrive in a stream. 
This problem can be formulated as follows.  The precise system state is comprised of $\Delta^d$ counters, namely $\sigma_{\vec{\mathbf{i}}}$ for $\mathbf{i}\in [0, \Delta)^d$, that are ``on the move''.   
Given a range $[\vec{\mathbf{l}},\vec{\mathbf{u}})$ at moment $t$, the goal is to approximately compute the sum of counter values in this range $C[\vec{\mathbf{l}},\vec{\mathbf{u}})\triangleq \sum_{\vec{\mathbf{i}}\in[\vec{\mathbf{l}},\vec{\mathbf{u}})} \sigma_{\vec{\mathbf{i}}}(t)$, where $\sigma_{\vec{\mathbf{i}}}(t)$ is the value of the counter $\sigma_{\vec{\mathbf{i}}}$ at moment $t$.   
A desirable solution to this problem in $d$D should satisfy three requirements (in which multiplicative terms related to the desired $(\epsilon, \delta)$ accuracy bound are ignored).
First, any range-sum query is answered in $O(\mathrm{polylog}(\Delta))$ time. 
Second, its space complexity is $O(\mathrm{polylog}(\Delta))$. 
Third, every point or range update to the system state is processed in $O(\mathrm{polylog}(\Delta))$ time. 
It has been a long-standing open question whether there is a solution to this problem that satisfies 
all three requirements when $d>1$.
For example, solutions producing exact answers (to the range queries)~\cite{ibtehazmdsegmenttree, propolyne, shiftsplit} all require $O(\Delta^d \log\Delta)$ space and hence do not satisfy the second requirement;
and Haar+ tree~\cite{haarplustree} works only on static data, and hence does not satisfy the third requirement.

In 1D, a solution that satisfies all three requirements (with $d=1$) was proposed in~\cite{onepasswavelet, dyahist}.
It involves 
1D-ERS computations on $4$-wise independent underlying RVs where the target distribution is either Gaussian or Rademacher, which are tackled using 
a DST-based (in~\cite{dyahist}) or a ECC-based (in~\cite{onepasswavelet}) 1D-ERS solution, respectively.  
As shown in~\cite{onepasswavelet, dyahist}, this range-sum query solution can be readily generalized to $d$D if the ERS computations above can be performed in $d$D. 
This gap is again filled by our $k$-wise ($k=4$) independence scheme for $d$D Gaussian-ERS, 
resulting in the first $d$D solution that satisfies all three requirements, all with $O(\log^d\Delta)$ (time or space) complexity (ignoring $\epsilon$ and $\delta$ terms).

In the resulting $d$D solution, we maintain $O(\log (1/\delta)/\epsilon^2)$ (independent instances of) sketches that each ``sketches'' the content (counter values) of the data cube.
Here we describe only one such sketch, which we denote as $A$, since these sketches are statistically and functionally identical.  
At any time $t$, $A(t)$ should track the current system state, namely ($\sigma_{\vec{\mathbf{i}}}(t)$)'s, as follows:  $A(t) \triangleq \sum_{\vec{\mathbf{i}}\in[0, \Delta)^d}\sigma_{\vec{\mathbf{i}}}(t) X_{\vec{\mathbf{i}}}$.
Here $X_{\vec{\mathbf{i}}}$ for $\mathbf{i}\in [0, \Delta)^d$ are (realizations of) a set of $\Delta^d$ $4$-wise independent standard Gaussian underlying RVs that have
one-to-one correspondences with the set of $\Delta^d$ counters as follows:  Each $X_{\vec{\mathbf{i}}}$ is associated with a counter $\sigma_{\vec{\mathbf{i}}}$.  
If we implement these $\Delta^d$ RVs using (an instance of) our $d$D Gaussian-ERS solution, then we can keep the value of $A(t)$ up-to-date, with a time complexity of $O(\log^d\Delta)$
per point or range update (to the system state).  
Then, given a query range $[\vec{\mathbf{l}},\vec{\mathbf{u}})$ at time $t$, we estimate the range-sum of counters 
$C[\vec{\mathbf{l}},\vec{\mathbf{u}})$ from the sketch $A(t)$ using $A(t)\cdot S[\vec{\mathbf{l}},\vec{\mathbf{u}})$ as the estimator.
These $O(\log (1/\delta)/\epsilon^2)$ estimators, one obtained from each sketch, are then combined to produce a final estimation that has the following accuracy guarantee (that is the 
same as in the 1D case).
With probability at least $1-\delta$, the final estimation deviates from the actual value of $C[\vec{\mathbf{l}},\vec{\mathbf{u}})$ by at most $\epsilon \sqrt{V[\vec{\mathbf{l}},\vec{\mathbf{u}})}\|\bm{\sigma}\|_2$, where $V[\vec{\mathbf{l}},\vec{\mathbf{u}}) \triangleq \prod_{j=1}^{d}(u_j - l_j)$ is the number of counters in the query range, and $ \|\bm{\sigma}(t)\|_2 \triangleq \left(\sum_{\vec{\mathbf{i}}\in[0, \Delta)^d} \sigma^2_{\vec{\mathbf{i}}}(t)\right)^{1/2}$ is the $L_2$ norm of the system state. 
Since each sketch uses an independent $d$D Gaussian-ERS scheme instance, our $d$D solution satisfies all three aforementioned requirements, all with $O(\log^d\Delta \log (1/\delta)/\epsilon^2)$ time and space complexity.

\subsection{A Closer Comparison with Related Work}

In this section, per referees' requests, we provide an in-depth comparison of this work with prior works on 1D-FAWT~\cite{onepasswavelet, dyahist}, 
on $d$D-FAWT~\cite{cormodewavelettracking}, and on 1D data 
cube~\cite{dyahist}.  

We start with explaining how the $d$D-FAWT solution proposed in~\cite{cormodewavelettracking} manages to avoid confronting the $d$D-ERS problem.  
The $d$D-FAWT solution~\cite{cormodewavelettracking} maintains ToW sketches for groups of wavelet coefficients in the wavelet domain.  
As explained earlier, each ToW sketch ``measures'' the $L_2$ norm (and hence the total energy by squaring) of such a group.     
By the property of HWT, each point or range update to the system state in the time domain translates into $O(\log^d \Delta)$ updates to the sketches the wavelet domain;
we also use this property in our solution to keep its time complexity below $O(\log^d \Delta)$ as shown in~\autoref{sec:waveletdd}.  
To solve the $d$D-FAWT using these sketches in the wavelet domain, we need only to identify the groups that are (hierarchical) ``$L_2$ heavy hitters''~\cite{cormodewavelettracking}.   
In~\cite{cormodewavelettracking}, a binary search tree built on these sketches is used to search for such ``$L_2$ heavy hitters'' in $O(\log \Delta\cdot \log\log \Delta)$ time.  
Since this $d$D-FAWT solution~\cite{cormodewavelettracking} does not involve computing the range sums of the Rademacher RVs underlying the ToW sketches, 
it does not need to formulate or solve any ERS problem.

As we will elaborate in Section~\ref{sec:dyasim}, our $d$D Gaussian-ERS solution works in the same way as the $d$D-FAWT solution proposed in~\cite{cormodewavelettracking},
by shifting the (representations of) input streams and the range queries from the time domain to the wavelet domain.  Hence, arguably had $d$D-FAWT solution proposed in~\cite{cormodewavelettracking} 
used the Gaussian ToW (GToW) instead of the ToW sketch, this shift would have resulted in a $d$D-FAWT solution containing the bulk of our $d$D Gaussian-ERS solution as an embedded module.  
However, such an embedded module is still ``two hops away'' from our $d$D Gaussian-ERS solution as follows.
First, since the objective of and the intuition behind this shift in~\cite{cormodewavelettracking} were to avoid rather than to solve the ERS problem, it would not be easy for the authors of~\cite{cormodewavelettracking} to realize that 
the embedded module can be extended to a standalone $d$D Gaussian-ERS solution.  
Second, without our aforementioned $k$-wise independence theory and construction, 
the embedded module does not yet guarantee $4$-wise independence among
underlying Gaussian RVs that is needed for $d$D-FAWT.  

On a related note, should we try to extend the 1D-FAWT solution proposed in~\cite{onepasswavelet}, which maintains the ToW sketches
in the time domain, to $d$D without the aforementioned Rademacher-by-Gaussian replacement, 
the underlying Rademacher RVs would have to be efficiently range-summable to keep the time complexity of each point or range update to the sketches low.   
However, this appears to be a tall order for now:  For $d>1$, no ECC-based Rademacher-ERS solution has ever been found as explained earlier, and a DST-based 
Rademacher-ERS solution is unlikely to exist, as we will show in \autoref{sec:infeasibility} and \autoref{sec:casenegative}.


A referee asked whether the 1D data cube solution proposed in~\cite{onepasswavelet, dyahist} can be extended to $d$D using the same aforementioned ERS avoidance strategy of 
maintaining the sketches in the wavelet domain as used in~\cite{cormodewavelettracking}.  In retrospect, this solution approach would work, but unlikely to be taken 
since it is counterintuitive and still ``two hops away'' (from the right solution) as explained above.  Indeed, authors of~\cite{onepasswavelet, dyahist} unsurprisingly took the much more intuitive approach of
maintaining sketches in the time domain and as a result had to confront the $d$D Gaussian or Rademacher-ERS problem as explained in \autoref{sec:datacube}.



Now we highlight a key difficulty that we believe has prevented authors of~\cite{onepasswavelet, dyahist, muthukrishnanhistogram} from solving the $d$D-ERS problem and 
extending their FAWT and data cube solutions from 1D to $d$D:  The Rademacher or Gaussian RVs underlying the sketches need to be both 4-wise independent
and efficiently range-summable, and conventional wisdom (until our work) has it that a magic hash function family is needed to achieve both.  
Authors of~\cite{onepasswavelet, muthukrishnanhistogram} tried to extend a magic hash function family, that induces such Rademacher RVs in 1D, to $d$D.
However, as explained earlier, a $d$D Rademacher-ERS solution is unlikely to exist.  Authors of~\cite{dyahist} proposed the 1D-DST that laid the foundation of this work 
and our prior work~\cite{1ddyadic}.  A key innovation of~\cite{dyahist} is that the 1D-ERS is achieved via a 1D-DST instead of a magic hash function.  
However, their DST-based 1D Gaussian-ERS solution still relies on a magic hash function, called Nisan's PRG (Pseudorandom Generator)~\cite{nisan}, to provide 4-wise independence among the underlying Gaussian RVs.
The use of Nisan's PRG~\cite{nisan} however restricts the applicability and the extensibility of the 1D-DST approach, since Nisan's PRG provides independence guarantees only for memory-constrained applications
such as data streaming~\cite{stabledist}.  It is also not clear whether the 1D-DST approach powered by Nisan's PRG can be extended to $d$D.  
In comparison, in our $d$D-ERS solutions, both $d$D-ERS and 4-wise independence are provided by the specially engineered $d$D-DST.
As a result, a magic hash function family is no longer needed, 
since the hash values produced by a hash function are no longer required to be efficiently range summable.


Finally, we state a key difference between this work and~\cite{onepasswavelet, dyahist, muthukrishnanhistogram, cormodewavelettracking} with respect to wavelets.  In this work, ERS is the end and wavelets is the means,
whereas in~\cite{onepasswavelet, dyahist, muthukrishnanhistogram} it is the other way around.  In~\cite{cormodewavelettracking}, wavelets is the end, but~\cite{cormodewavelettracking} cleverly avoids using 
ERS as the means as just explained.

\section{Our Solution to $\bf{d}$D Gaussian-ERS}\label{sec:dyasim}
In this section, we describe our $d$D Gaussian-ERS solution that answers a range-sum query in $O(\log^d\Delta)$ time.  
To explain this solution with best clarity, for now we require it to provide the aforementioned ideal guarantee that the $\Delta^d$
underlying RVs are {\it mutually independent}, with the understanding that this requirement affects only the space complexity of 
our solution.  In the next section, this requirement will be relaxed to these RVs being $k$-wise independent, and as a result, 
the space complexity of our solution is reduced to $O(\log^d\Delta)$.  

Our solution can be summarized as follows.  
Let $\vec{\mathbf{x}}$ denote the $\Delta^d$ underlying standard Gaussian RVs, namely $X_{\vec{\mathbf{i}}}$ for $\vec{\mathbf{i}} \in [0, \Delta)^d$, arranged (in the dictionary order of $\vec{\mathbf{i}}$) 
into a $\Delta^d$-dimensional vector.  Then, after the $d$D Haar wavelet transform (HWT) is performed on $\vec{\mathbf{x}}$, 
we obtain another $\Delta^d$-dimensional vector $\vec{\mathbf{w}}$ whose scalars are the 
HWT coefficients of $\vec{\mathbf{x}}$.  Our solution builds on the following two observations.
The first observation is that scalars in $\vec{\mathbf{x}}$ are i.i.d. standard Gaussian RVs if and only if scalars in $\vec{\mathbf{w}}$ are (see \Cref{lem:correctness}).  
The second observation is that the answer to any $d$D range-sum query 
can be expressed as a weighted sum of $O(\log^d\Delta)$ scalars (HWT coefficients) in $\vec{\mathbf{w}}$ (see~\Cref{lem:log2delta}).  
Our algorithm
is simply to {\it generate and remember} only these $O(\log^d\Delta)$ HWT coefficients (that participate in this range-sum query).  
Our solution satisfies the 
correct distribution requirement (with mutual independence guarantee) by the first observation.
Since the first observation is true only when the target distribution is Gaussian,
this HWT-based solution does not work for any other target distribution. 



In the following, we first introduce the concept of the dyadic simulation tree (DST) in 1D in~\autoref{sec:dst1D}.
Then, we show that 1D DST is equivalent to 1D HWT in the Gaussian case and present our HWT-based Gaussian-ERS algorithm for 1D, in~\autoref{sec:wavelet1d}.
Finally, we describe our HWT-based Gaussian-ERS algorithms for 2D and $d$D in~\autoref{sec:2dwavelet} and~\autoref{sec:waveletdd}, respectively.

\subsection{A Brief Introduction to DST}\label{sec:dst1D}

In this section, we briefly introduce the 
concept of the DST, which as mentioned earlier was proposed in~\cite{1ddyadic} as a general solution approach to the one-dimensional ($1$D) ERS problems for  
arbitrary target distributions.

\begin{figure*}
	\centering
	\begin{subfigure}[t]{0.27\textwidth}
		\centering
		\includegraphics[width=\textwidth]{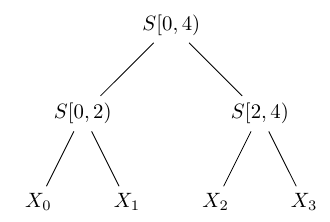}
			\centering
		\caption{A general DST.}\label{fig:1ddst}
	\end{subfigure}
		\begin{subfigure}[t]{0.68\textwidth}
		\centering
		\includegraphics[width=\textwidth]{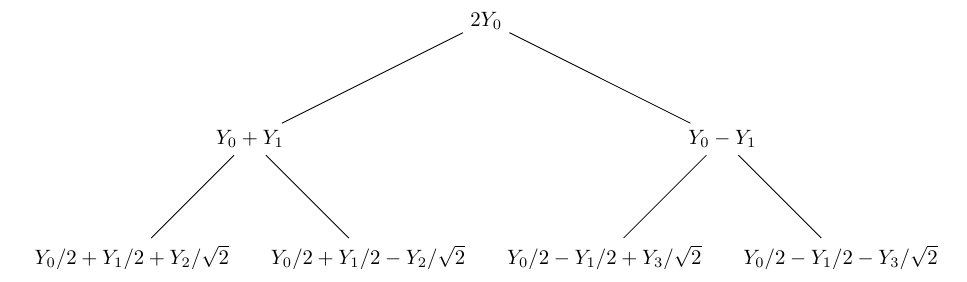}
			\centering
		\caption{A Gaussian-DST.}\label{fig:gaudst}
	\end{subfigure}
\caption{Illustrations of a general DST and a Gaussian-DST with $\Delta=4$.}
\end{figure*}

We say that $[l, u)$ is a 1D \emph{dyadic range} if there exist integers $j \ge 0$ and $m \ge 0$ such that $l = j \cdot 2^m$ and $u = (j+1) \cdot 2^m$. 
We call the sum on a dyadic range a \emph{dyadic range-sum}.  Note that any underlying RV $X_i$ is a dyadic range-sum (on the dyadic range $[i, i+1)$).
Let each underlying RV $X_i$ 
have standard Gaussian distribution $\mathcal{N}(0, 1)$. 
In the following, we focus on how to compute a dyadic range-sum, since any (general) 1D range can be ``pieced together'' using at most 
$2\log_2 \Delta$ dyadic ranges~\cite{rusu-jour}.
We illustrate the process of computing dyadic range-sums using a ``small universe'' example (with $\Delta= 4$) shown in \autoref{fig:1ddst}.  To begin with, the total sum of the universe $S[0, 4)$ sitting at the root of the tree is generated directly from its distribution $\mathcal{N}(0, 4)$. Then, $S[0, 4)$ is split into two children, the half-range-sums $S[0, 2)$ and $S[2, 4)$, such that RVs $S[0, 2)$ and $S[2, 4)$ sum up to $S[0, 4)$, are (mutually) independent, and each has distribution $\mathcal{N}(0, 2)$. This is done by generating the RVs $S[0, 2)$ and $S[2, 4)$ from a 
conditional (upon $S[0, 4)$) distribution that will be specified shortly. Afterwards, $S[0, 2)$ is split in a similar way into two i.i.d. underlying RVs $X_0$ and $X_1$, and so is $S[2, 4 )$ (into
$X_2$ and $X_3$).  As shown in~\autoref{fig:1ddst}, the four underlying RVs are the leaves of the DST.

We now specify the aforementioned conditional distribution used for each split. Suppose the range-sum to split consists of $2n$ underlying RVs, and that its value is  equal to $z$. The lower half-range-sum $S_l$ (the left child in~\autoref{fig:1ddst}) is generated from the following conditional pdf (or pmf): 
\begin{equation}\label{eq:dstsplit}
	f(x\mid z) = \phi_n(x)\phi_n(z-x)/\phi_{2n}(z),
\end{equation}
where $\phi_n(\cdot)$ is the pdf (or pmf) of $X^{*n}$, the $n^{th}$ convolution power of the target distribution, and $\phi_{2n}(\cdot)$ is the pdf (or pmf) of $X^{*2n}$. Then, the upper half-range-sum (the right child) is defined as $S_u\triangleq z-S_l$.  
It was shown in~\cite{1ddyadic} that splitting a (parent) RV using this conditional distribution guarantees that the two resulting RVs $S_l$ and $S_u$ are i.i.d.   
This guarantee holds regardless of the target distribution. 
However, {\it computationally efficient} procedures for generating an RV $S_l$ with distribution $f(x\mid z)$ 
are found only when the target distribution is one of the few ``nice''  distributions:   
Gaussian, Cauchy, and Rademacher as shown in~\cite{1ddyadic}, and Poisson as shown in~\autoref{sec:2ddyasim}.
  

Among them, Gaussian distribution has a nice property that an RV $S_l$ with distribution $f(x\mid z)$ can be generated as a linear combination of $z$ and a ``fresh'' standard Gaussian RV $Y$ as $S_l\triangleq z/2 + \sqrt{n/2}\cdot Y$, 
since if we plug Gaussian pdfs  $\phi_n(\cdot)$ and $\phi_{2n}(\cdot)$ into~(\ref{eq:dstsplit}), $f(x\mid z)$ is precisely the pdf of $\mathcal{N}(z/2, n/2)$.  Here, $Y$ being ``fresh'' means it is independent of all other RVs.


This linearly decomposable property has a pleasant consequence that \emph{every} dyadic range-sum generated by this $1$D Gaussian-DST can be recursively decomposed to a linear combination of some i.i.d. standard Gaussian RVs, as illustrated in~\autoref{fig:gaudst}. In this example, let $Y_0, Y_1, Y_2$ and $Y_3$ be four i.i.d. standard Gaussian RVs. 
The total sum of the universe $S[0, 4)$ is written as $2Y_0$, because they have the same distribution $\mathcal{N}(0, 4)$. Then, it is split into two half-range-sums $S[0,2)\triangleq Y_0 + Y_1$ and $S[2, 4)\triangleq Y_0 - Y_1$ using the linear decomposition above with $z=2Y_0$ and a fresh RV $Y_1$. Finally, $S[0,2)$ and $S[2, 4)$ are similarly split into the four underlying RVs using fresh RVs $Y_2$ and $Y_3$, respectively. 


\subsection{HWT Representation of 1D Gaussian-DST}\label{sec:wavelet1d}


In this section, we show that when the target distribution is Gaussian, 
a DST is mathematically equivalent to a Haar wavelet transform (HWT) in the 1D case.  We will also show that this equivalence carries over to higher dimensions.  
Note that this equivalence does not apply to 
any target distribution other than Gaussian, and hence the HWT representation cannot replace 
the role of DST in general.
In the following, we describe in~\autoref{sec:1dalgorithm} our HWT-based 1D Gaussian-ERS solution that has $O(\log\Delta)$ time complexity, after making some mathematical preparations in~\autoref{sec:1dmath}.

\subsubsection{Mathematical Preliminaries}\label{sec:1dmath}
It is not hard to verify that, if we apply HWT (to be specified soon) to the four underlying RVs shown in~\autoref{fig:gaudst}, namely
$X_0 = Y_0/2 + Y_1/2 + Y_2/\sqrt{2}$, $X_1 = Y_0/2 + Y_1/2 - Y_2/\sqrt{2}$, $X_2 = Y_0/2 - Y_1/2 + Y_3/\sqrt{2}$, and $X_3 = Y_0/2 - Y_1/2 - Y_3/\sqrt{2}$, then the four HWT coefficients we obtain are precisely $Y_0, Y_1, Y_2, Y_3$, respectively.  In other words, we have  $\vec{\mathbf{y}} = H_4 \vec{\mathbf{x}}$, where
$\vec{\mathbf{x}} \triangleq (X_0, X_1, X_2, X_3)^T$, $\vec{\mathbf{y}} \triangleq (Y_0, Y_1, Y_2, Y_3)^T$, and $H_4$ is the $4\times 4$ Haar matrix $H_4$.  
This example is illustrated as a matrix-vector multiplication in~\autoref{fig:dhwt1d}.

\begin{figure}
	\centering\includegraphics[width=0.7\textwidth]{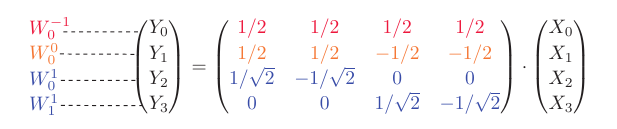}
	\caption{An illustration of the  HWT formula $\vec{\mathbf{y}} = H_4 \vec{\mathbf{x}}  $.
	}\label{fig:dhwt1d}
\end{figure}

The above example in which $\Delta = 4$ can be generalized to an arbitrary universe size $\Delta$ (that is a power of $2$) as follows.
In general, HWT is defined as $ \vec{\mathbf{w}} = H_{\Delta} \vec{\mathbf{x}}$, where $\vec{\mathbf{w}}$ and $\vec{\mathbf{x}}$ are both $\Delta$-dimensional vectors, and
$H_{\Delta}$ is a $\Delta\times \Delta$ Haar matrix.
To simplify notations, we drop the subscript $\Delta$ in the sequel.
In wavelet terms, $\vec{\mathbf{x}}$ is called a discrete signal vector and $\vec{\mathbf{w}}$ is called the \emph{HWT coefficient} vector.  
Clearly, the $i^{th}$ HWT coefficient is the inner product between $\vec{\mathbf{x}}$ and the $i^{th}$ row of $H$, for $i = 0, 1,  \cdots, \Delta-1 $.
In the wavelet theory, we index each HWT coefficient as $W^m_j$ (instead of $W_i$) for $m=-1, 0, 1, \cdots, \log_2\Delta -1$ and $j =0, 1, \cdots, 2^{m^+}-1$ (where $m^+\triangleq \max\{0, m\}$) in the dictionary order of $(m, j)$, and refer to the corresponding 
row (transposed into a column vector)  in $H$  that computes $W^m_j$  as the \emph{HWT vector} $\vec{\bm{\psi}}^m_j$.  Hence we have $W^m_j \triangleq \langle\vec{\mathbf{x}}, \vec{\bm{\psi}}^m_j\rangle$ by definition.
In wavelet terms, parameter $m$ is called scale and parameter $j$ is called  location.
In~\autoref{fig:dhwt1d}, the 4 HWT coefficients and 4 HWT vectors from top to bottom are on 3 different scales ($-1, 0$, and $1$) and are ``assigned'' 3 different colors accordingly.

We define the \emph{indicator vector} of a 1D range $R$, denoted as $\mathbbm{1}_R$, as a $\Delta$-dimensional 0-1 vector, the $i^{th}$ scalar of which takes value $1$ if $i\in R$ and $0$ otherwise, 
for $i = 0, 1, \cdots, \Delta-1$.  Throughout this paper, the indicator vectors are the only vectors that are not written in boldface with a rightward arrow on the top.  
We now specify the HWT vectors. Every HWT vector $\vec{\bm{\psi}}^m_j$ is normalized such that $\| \vec{\bm{\psi}}^m_j\|_2 = 1$. 
The first HWT vector $\vec{\bm{\psi}}^{-1}_{0}\triangleq \Delta^{-1/2}\cdot \mathbbm{1}_{[0,\Delta)}$ is special:  Its corresponding coefficient $W^{-1}_0$ reflects the scaled 
(by $\Delta^{-1/2}$) range-sum of the entire universe, whereas
every other HWT coefficient is the (scaled) 
difference of two range-sums.  
Every other HWT vector $\vec{\bm{\psi}}^m_j$, for $m =0, 1, \cdots, \log_2\Delta-1$ and $j=0,1, \cdots, 2^m-1$, 
corresponds to the dyadic range $I^m_j \triangleq [j\Delta/2^m, (j+1)\Delta/2^m)$
in the sense the latter serves as the support of the former: $\vec{\bm{\psi}}^m_j$ is defined by setting the first half of $I^m_j$ to the value $\sqrt{2^{m}/\Delta}$, the second half of $I^m_j$ to the value $-\sqrt{2^{m}/\Delta}$, and the rest of the universe $[0, \Delta)\setminus I^m_j$ to the value $0$. Note that $\vec{\bm{\psi}}^m_j$ has the same number of scalars with value $\sqrt{2^{m}/\Delta}$ as those with value $-\sqrt{2^{m}/\Delta}$, so $\langle\vec{\bm{\psi}}^m_j, \mathbbm{1}_{I^m_j}\rangle= 0$.
From the definition above, $H$ is known to be orthonormal~\cite{nievergeltwavelet}, so the following theorem can be applied to it.
\begin{theorem}[\cite{laubmatrixbook}]\label{lem:orthoh}
Let $M$ be an $n\times n$ matrix.  If $M$ is orthonormal, then it has the following two properties: 
\begin{enumerate}
	\item $M^T = M^{-1}$, and $M^T$ is also orthonormal.
	\item Given any two $n$-dimensional vectors $\vec{\mathbf{x}}, \vec{\mathbf{y}}$, we have $\langle \vec{\mathbf{x}}, \vec{\mathbf{y}} \rangle = \langle M\vec{\mathbf{x}}, M\vec{\mathbf{y}} \rangle $.
\end{enumerate}
\end{theorem}

Let $\vec{\mathbf{w}}$ be a $\Delta$-dimensional vector of i.i.d. standard Gaussian RVs.
We mathematically define the vector of underlying RVs $\vec{\mathbf{x}} = (X_0, X_1, \cdots, X_{\Delta-1})^T$ as $\vec{\mathbf{x}}\triangleq H^T \vec{\mathbf{w}}$. Hence, we have $\vec{\mathbf{w}} = H\vec{\mathbf{x}}$ by the first property  in~\Cref{lem:orthoh}.
The underlying RVs defined this way are i.i.d. standard Gaussian, by the following theorem.

\begin{theorem}[Proposition 3.3.2 in~\cite{vershyninprobability}]\label{lem:correctness}
	Let $\vec{\mathbf{x}} = M\vec{\mathbf{w}}$ where $M$ is an orthonormal matrix.  Then $\vec{\mathbf{x}}$ is a vector of i.i.d. standard Gaussian RVs if and only if $\vec{\mathbf{w}}$ is.
\end{theorem}


\subsubsection{Our HWT-based Algorithm for $1$D-ERS}\label{sec:1dalgorithm}
Given any range $[l, u)$, we compute its range-sum $S[l, u)$ as 
$\langle \vec{\mathbf{w}}, H\mathbbm{1}_{[l, u)}\rangle$, which is the sum of the HWT coefficients in $\vec{\mathbf{w}}$ weighted by the scalars in $H\mathbbm{1}_{[l, u)}$.
This weighted sum can be computed in $O(\log\Delta)$ time, because,
by~\Cref{cor:logdelta}, the $\Delta$-dimensional vector $H\mathbbm{1}_{[l, u)}$ contains only $O(\log\Delta)$ nonzero
scalars (weights), and by~\Cref{rem:o1time}, for each such scalar, its index can be located and its value computed in $O(1)$ time.   
We refer to the $O(\log\Delta)$ corresponding scalars in $\vec{\mathbf{w}}$ whose weights are nonzero as \emph{participating HWT coefficients} in the sequel.


To provide the aforementioned ideal guarantee of \emph{mutual independence} (among the $\Delta$ underlying RVs), for each such participating HWT coefficient (which is a standard Gaussian RV),
we generate the RV and remember its realization (in memory) if it has never been generated before (say for answering an earlier range-sum query), 
or retrieve its realization from memory otherwise.   The space complexity of this algorithm is $O(\min\{T\log \Delta , \Delta\})$, since each of the $T$ range-sum queries involves $O(\log\Delta)$ participating HWT coefficients.
This algorithm satisfies the aforementioned consistency requirement, because $\langle \vec{\mathbf{w}}, H\mathbbm{1}_{[l, u)}\rangle = \langle H\vec{\mathbf{x}}, H\mathbbm{1}_{[l, u)}\rangle = \langle \vec{\mathbf{x}}, \mathbbm{1}_{[l, u)}\rangle = X_{l} + X_{l+1} + \cdots + X_{u-1}$.  The second equation above is by the second property  in~\Cref{lem:orthoh}.

\begin{theorem}\label{cor:logdelta}
	Given any range $[l, u) \subseteq [0, \Delta)$, $H \mathbbm{1}_{[l, u)}$ contains at most $2\log_2\Delta +2$ nonzero scalars.
\end{theorem}



\Cref{cor:logdelta} is a straightforward corollary of~\autoref{lem:complexity1d}, since $H$ has only $\log_2\Delta+1$ scales.



\begin{lemma}\label{lem:complexity1d}
	Given any range $[l, u)\subseteq [0, \Delta)$, $H \mathbbm{1}_{[l, u)}$ contains at most $2$ nonzero scalars on each scale.
\end{lemma}
\begin{proof}
On scale $m=-1$, there is only one HWT coefficient anyway, so the claim trivially holds.
We next prove the claim for any fixed $m\geq 0$.  
For each HWT vector $\vec{\bm{\psi}}^m_j$, $j= 0, 1, \cdots, 2^m-1$, we denote the corresponding HWT coefficient as $r^m_j \triangleq \langle \vec{\bm{\psi}}^m_j, \mathbbm{1}_{[l, u)}\rangle$.
It is not hard to verify that the relationship between the range $[l, u)$ and the dyadic range $I^m_j$ must be one of the following three cases.
	\begin{enumerate}
		\item $I^m_j$ and $[l, u)$ are disjoint. In this case, $r^m_j = 0$.
		\item $I^m_j \subseteq [l, u)$. In this case,  $r^m_j = \langle \vec{\bm{\psi}}^m_j, \mathbbm{1}_{I^m_j}\rangle = 0$ as explained in the second last sentence above~\Cref{lem:orthoh}. 
		\item \label{case:nonzero} Otherwise, $I^m_j$ partially intersects $[l, u)$. This case may happen only to at most two ($I^m_j$)'s: the one that covers $l$ and the one that covers $u-1$. In this case, $r^m_j$ can be nonzero.
	\end{enumerate}
\end{proof}
\begin{remark}\label{rem:o1time}
Each scalar $r^m_j$ (in $H \mathbbm{1}_{[l, u)}$) that may be nonzero can be identified and computed in $O(1)$ time as follows. 
Note $r^m_j$ may be nonzero only in the case~(\ref{case:nonzero}) above, in which  $j$ is equal to either 
$\lfloor l2^m/\Delta\rfloor$ or $\lfloor (u-1)2^m/\Delta\rfloor$.  As a result, if $r^m_j\not=0$, its value can be computed in two steps~\cite{propolyne}.
First, intersect $[l, u)$ with the first half and the second half of $I^m_j$, respectively.
Second, scale the size of the first intersection minus the size of the second by  $\sqrt{2^{m}/\Delta}$, as was explained by the third last sentence above~\Cref{lem:orthoh}. 
\end{remark}

The following lemma is a special case of~\Cref{lem:complexity1d} where $l = u-1$. 
This lemma holds, because in case~(\ref{case:nonzero}) above, for $m = -1, 0, 1, \cdots, \log_2\Delta-1$, there exists a  unique dyadic interval $I^m_j$ that covers $l = u-1$ (namely, the one with $j=\lfloor l2^m/\Delta\rfloor$).


\begin{lemma}\label{lem:onehot}
	Given any $l\in [0, \Delta)$, $H \mathbbm{1}_{\{l\}}$ has exactly one nonzero scalar on each scale.  
\end{lemma}

\subsection{Range-Summable Gaussian RVs in 2D}\label{sec:2dwavelet}

In the following, we describe in~\autoref{sec:2dalgo} our 2D Gaussian-ERS solution that has $O(\log^2 \Delta)$ time complexity, after making some mathematical preparations in~\autoref{sec:2dmath}.
\subsubsection{Mathematical Preliminaries}\label{sec:2dmath}
Like in the 1D case, our 2D Gaussian-ERS solution builds on the 2D-HWT $\vec{\mathbf{w}} = H^{\otimes 2}\vec{\mathbf{x}}$.  
Here the vector $\vec{\mathbf{x}}$ is comprised of the $\Delta^2$ underlying RVs $X_{\vec{\mathbf{i}}}$ for $\vec{\mathbf{i}}\in[0, \Delta)^2$, listed in the dictionary order; and the vector $\vec{\mathbf{w}}$ is comprised of the resulting $\Delta^2$ 
2D-HWT coefficients.  The $\Delta^2 \times \Delta^2$ 2D-HWT matrix $H^{\otimes 2}$ is the self {\it Kronecker product} (defined next) of the $\Delta\times\Delta$ 1D-HWT matrix $H$.



\begin{definition}\label{def:kronecker}
	Let $A$ be a $p\times q$ matrix and $B$ be a $t\times v$ matrix. Then their Kronecker product $A\otimes B$ is the following $pt\times qv$ matrix.
	\[A\otimes B \triangleq \begin{pmatrix}
		a_{11} B & \cdots & a_{1q}B \\
		\vdots & \ddots & \vdots\\
		a_{p1} B & \cdots & a_{pq} B
	\end{pmatrix}.\]
\end{definition}

We now state two theorems concerning the Kronecker product.

\begin{theorem}[Theorem 13.3 in~\cite{laubmatrixbook}]\label{th:mixedproduct}
	Let $P, Q, R, T$ be four matrices such that the matrix products $P\cdot R$ and $Q\cdot T$ are well-defined. Then $(P\otimes Q)\cdot (R\otimes T) = (P\cdot R)\otimes (Q\cdot T)$.
\end{theorem}

\begin{theorem}[Corollary 13.8 in~\cite{laubmatrixbook}]\label{lem:kroneckerorthonormal}
	The Kronecker product of two orthonormal matrices is also orthonormal.
\end{theorem}

Now we describe the $\Delta^2$ 2D HWT coefficients and the order in which they are listed in $\vec{\mathbf{w}}$.
Recall that in the 1D case, each HWT coefficient takes the form $W^m_j$, where $m$ is the scale, and the $j$ is the location.  
In the 2D case, each dimension has its own pair of scale and location parameters that is independent of the other dimension.  
For convenience of presentation, we refer to these two dimensions as vertical (the first) and horizontal (the second), respectively.
We denote the vertical scale and location pair as $m_1$ and $j_1$, and the horizontal pair as $m_2$ and $j_2$.  
Each HWT coefficient takes the form $W^{m_1, m_2}_{j_1, j_2}$.  
In the 2D case, there are $(\log_2\Delta + 1)^2$ scales, namely $(m_1, m_2)$ for $m_1, m_2 = -1, 0, 1, \cdots, \log_2\Delta-1$.
At each scale $(m_1, m_2)$, there are $n_{m_1, m_2} \triangleq 2^{m_1^++m_2^+}$ locations, namely $(j_1, j_2)$ for $j_1 = 0, 1, \cdots, 2^{m_1^+}-1$ and $j_2 = 0, 1, \cdots, 2^{m_2^+}-1$.

We now give a 2D example in which $\Delta = 4$.
In this 2D example, there are $\Delta^2 = 16$ HWT coefficients.  To facilitate the ``color coding'' of different scales, we arrange the 16 HWT coefficients 
into a $4\times 4$ matrix $\textbf{W}$ shown in~\autoref{fig:colors}.
$\textbf{W}$ is the only matrix that we write in boldface in order to better distinguish it from its scalars ($W^{m_1, m_2}_{j_1, j_2}$)'s. 
\autoref{fig:colors} contains three differently colored rows of heights 1, 1, and 2 respectively, that correspond to vertical scales $m_1 = -1, 0, 1$ respectively, and contains three differently colored columns that correspond to the three horizontal scales.
Their ``Cartesian product'' contains 9 ``color cells'' that correspond to the 9 different scales (values of $(m_1, m_2)$).   For example, the cell colored in pink corresponds to scale $(1, 1)$
and contains 4 HWT coefficients $W^{1,1}_{0,0}$, $W^{1,1}_{0,1}$, $W^{1,1}_{1,0}$, $W^{1,1}_{1,1}$.
The vector $\vec{\mathbf{w}}$ is defined from $\mathbf{W}$ by flattening its 16 scalars in the row-major order, as shown at the bottom of~\autoref{fig:colors}.

\begin{figure*}[!ht]
	\centering\includegraphics[width=\textwidth]{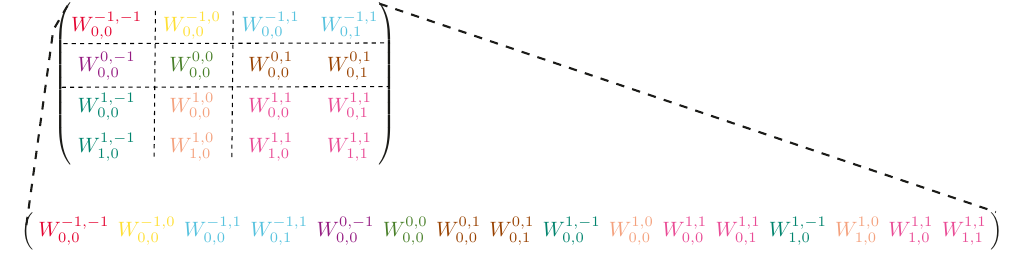}
	\caption{The 2D-HWT coefficients, arranged both as a matrix $\mathbf{W}$ and as a flattened vector $\vec{\mathbf{w}}^T$.}\label{fig:colors}
\end{figure*}
Like in the 1D case, let $\vec{\mathbf{w}}$ be a vector of $\Delta^2$ i.i.d. standard Gaussian RVs. 
As explained earlier, the vector $\vec{\mathbf{x}}$ of $\Delta^2$ underlying RVs are \emph{mathematically} defined as $\vec{\mathbf{x}}\triangleq (H^{\otimes 2})^T \vec{\mathbf{w}}$. 
The RVs in $\vec{\mathbf{x}}$ are i.i.d. standard Gaussian by~\Cref{lem:correctness}, because $H^{\otimes 2}$ is an orthonormal matrix by~\Cref{lem:kroneckerorthonormal}.
\subsubsection{Our HWT-Based Algorithm for $2$D-ERS}\label{sec:2dalgo}



Our 2D-ERS algorithm (that guarantees mutual independence among the underlying RVs) is similar to the 1D-ERS algorithm described earlier. 
Given any 2D range $[\vec{\mathbf{l}}, \vec{\mathbf{u}}) \triangleq [l_1, u_1) \times [l_2, u_2)$, where $\vec{\mathbf{l}} = (l_1, l_2)^T$ and $\vec{\mathbf{u}} = (u_1, u_2)^T$, 
we compute its range-sum $S[\vec{\mathbf{l}}, \vec{\mathbf{u}})$ as $\langle \vec{\mathbf{w}}, H^{\otimes 2}\mathbbm{1}_{[\vec{\mathbf{l}}, \vec{\mathbf{u}})}\rangle$. 
Here the 2D indicator vector $\mathbbm{1}_{[\vec{\mathbf{l}}, \vec{\mathbf{u}})}$ is defined as the result of flattening the following $\Delta\times\Delta$ matrix in  row-major order: 
For $\vec{\mathbf{i}}\in [0, \Delta)^2$, the $\vec{\mathbf{i}}^{th}$ scalar in the matrix takes value $1$ if $\vec{\mathbf{i}}\in [\vec{\mathbf{l}}, \vec{\mathbf{u}})$ and takes value $0$ otherwise.
	This return value $\langle \vec{\mathbf{w}}, H^{\otimes 2}\mathbbm{1}_{[\vec{\mathbf{l}}, \vec{\mathbf{u}})}\rangle$ can be computed in $O(\log^2\Delta)$ time, since 
it involves generating, and computing the weighted sum of, $O(\log^2\Delta)$ participating HWT coefficients according to~\Cref{lem:log2delta}.   The space complexity is $O(\min\{T\log^2 \Delta, \Delta^2\})$
for remembering the realizations of the $O(\log^2 \Delta)$ participating HWT coefficients (per query) like that explained earlier in the 1D case.  
Our 2D-ERS algorithm meets the consistency requirement, 
because $\langle \vec{\mathbf{w}}, H^{\otimes 2}\mathbbm{1}_{[\vec{\mathbf{l}}, \vec{\mathbf{u}})}\rangle =  \langle H^{\otimes 2}\vec{\mathbf{x}}, H^{\otimes 2}\mathbbm{1}_{[\vec{\mathbf{l}}, \vec{\mathbf{u}})}\rangle  =  \langle \vec{\mathbf{x}}, \mathbbm{1}_{[\vec{\mathbf{l}}, \vec{\mathbf{u}})}\rangle = \sum _{(i_1, i_2) \in [\vec{\mathbf{l}}, \vec{\mathbf{u}})} X_{i_1, i_2}$. 

\begin{lemma}\label{lem:log2delta}
	For any 2D range $[\vec{\mathbf{l}}, \vec{\mathbf{u}}) \subseteq [0, \Delta)^2$, $H^{\otimes 2}\mathbbm{1}_{[\vec{\mathbf{l}}, \vec{\mathbf{u}})}$ has $O(\log^2 \Delta)$ nonzero scalars.
\end{lemma}

\begin{proof}
	It is not hard to verify $\mathbbm{1}_{[\vec{\mathbf{l}}, \vec{\mathbf{u}})}  = \mathbbm{1}_{[l_1, u_1)} \otimes \mathbbm{1}_{[l_2, u_2)}$.
	By~\Cref{th:mixedproduct}, $H^{\otimes 2}\mathbbm{1}_{[\vec{\mathbf{l}}, \vec{\mathbf{u}})} = (H \otimes H) \cdot (\mathbbm{1}_{[l_1, u_1)} \otimes \mathbbm{1}_{[l_2, u_2)})= (H\mathbbm{1}_{[l_1, u_1)}) \otimes (H\mathbbm{1}_{[l_2, u_2)})$. By~\Cref{cor:logdelta}, both $H \mathbbm{1}_{[l_1, u_1)}$ and $H\mathbbm{1}_{[l_2, u_2)}$ have $O(\log \Delta)$ nonzero scalars, so their Kronecker product has $O(\log^2 \Delta)$ nonzero scalars.
\end{proof}

\subsection{Generalization to Higher Dimensions}\label{sec:waveletdd}

Our HWT-based Gaussian-ERS solution, just like HWT itself, can be naturally generalized to higher dimensions as follows.
In dimension $d > 2$, we continue to have the inverse HWT formula $\vec{\mathbf{x}}\triangleq M^T \vec{\mathbf{w}}$, 
where $\vec{\mathbf{x}}$ is the vector of $\Delta^d$ underlying RVs (arranged in dictionary order of $\vec{\mathbf{i}}$), $\vec{\mathbf{w}}$ is the vector of their HWT coefficients (that are
i.i.d. standard Gaussian RVs), 
and $M$ is the $\Delta^d\times\Delta^d$ HWT matrix in $d$D.  Here $M\triangleq \underbrace{H\otimes \cdots \otimes H}_{d}$, where $H$ is the 1D Haar matrix described above.  Since $M$ is orthonormal by \Cref{lem:kroneckerorthonormal}, the RVs in $\vec{\mathbf{x}}$ are i.i.d. standard Gaussian by~\Cref{lem:correctness}. 

In our $d$D-ERS algorithm (that guarantees mutual independence among the underlying RVs), 
given a $d$D range $[\vec{\mathbf{l}}, \vec{\mathbf{u}})\triangleq [l_1, u_1)\times [l_2, u_2)  \times \cdots\times [l_d, u_d) $, 
its range-sum $S[\vec{\mathbf{l}}, \vec{\mathbf{u}}) $ can be computed as $\langle \vec{\mathbf{w}}, M\mathbbm{1}_{[\vec{\mathbf{l}}, \vec{\mathbf{u}})}\rangle$, because $\langle \vec{\mathbf{w}}, M\mathbbm{1}_{[\vec{\mathbf{l}}, \vec{\mathbf{u}})}\rangle = \langle M\vec{\mathbf{x}}, M\mathbbm{1}_{[\vec{\mathbf{l}}, \vec{\mathbf{u}})}\rangle = \langle \vec{\mathbf{x}}, \mathbbm{1}_{[\vec{\mathbf{l}}, \vec{\mathbf{u}})}\rangle  = \sum_{\vec{\mathbf{i}} \in [\vec{\mathbf{l}}, \vec{\mathbf{u}})} X_{\vec{\mathbf{i}}}$.  
The weighted sum
$\langle \vec{\mathbf{w}}, M\mathbbm{1}_{[\vec{\mathbf{l}}, \vec{\mathbf{u}})}\rangle$ can be computed in $O(\log^d\Delta)$ time, because the weight vector $M\mathbbm{1}_{[\vec{\mathbf{l}}, \vec{\mathbf{u}})} =  M\cdot(\mathbbm{1}_{[l_1, u_1)}\otimes \mathbbm{1}_{[l_2, u_2)}\otimes \cdots\otimes \mathbbm{1}_{[l_d, u_d)}) = (H\mathbbm{1}_{[l_1, u_1)}) \otimes (H\mathbbm{1}_{[l_2, u_2)}) \otimes \cdots \otimes (H\mathbbm{1}_{[l_d, u_d)})$ has only $O(\log^d \Delta)$ nonzero scalars (weights) by~\Cref{cor:logdelta} and 
the property of Kronecker product.  Hence, we need to generate and remember only $O(\log^d \Delta)$ corresponding
participating HWT coefficients.  
As a result, our $d$D-ERS algorithm has $O(\min\{T\log^d \Delta, \Delta^d\})$ space complexity.
\section{$\bf{k}$-wise Independence Theory}\label{sec:kwise}

	


In this section, in all subsequent paragraphs, we assume $d=2$ ($2$D) for notational simplicity.  
All our statements and proofs can be readily generalized to higher dimensions. 
Recall that, for guaranteeing mutual independence among the $\Delta^d$ underlying RVs, 
our HWT-based $d$D Gaussian-ERS needs to \emph{remember} (the realization of) every participating HWT coefficient that was generated for
answering a past range-sum query, which can lead to high storage overhead when the number of queries $T$ is large.
In this section
we propose a $k$-wise independence theory and 
scheme that guarantees that the $\Delta^d$ underlying Gaussian RVs are $k$-wise independent.
It does so by using $O(\log^d \Delta)$ $k$-wise independent hash functions (described next)
instead.
This scheme has the same time complexity of $O(\log^d \Delta)$ as the idealized Gaussian-ERS solution, and
a much smaller space complexity of $O(\log^{d} \Delta)$, for storing the seeds of 
$O(\log^d \Delta)$ $k$-wise independent hash functions.
This scheme significantly extends its 1D version proposed in~\cite{1ddyadic}.  
Finally, we note this scheme works also for our Poisson-ERS solution.  We however will not explain how it works in this paper, since doing so would involve drilling down to 
the messy and lengthy detail of the Cartesian product of $d>1$ DSTs (since we cannot use the relatively clean and simple $d$D HWT in the Poisson case).

A $k$-wise independent hash function $h(\cdot)$ has the following property:  Given an arbitrary set of $k$ 
distinct keys $i_1, i_2, \cdots, i_k$, their hash values $h(i_1), h(i_2), \cdots, h(i_k)$ are independent.   
Such hash functions are very computationally efficient when $k$ is a small number
such as $k=2$ (roughly 2 nanoseconds per hash) and $k = 4$ (several nanoseconds per 
hash)~\cite{univhash,tabulation4wise,tabulationpower}.  
Typically, the hash values are (uniform random) integers. We can map them to Gaussian RVs using a deterministic function $g(\cdot)$ such as the Box-Muller transform~\cite{boxmuller}.

Recall (from~\autoref{fig:colors}) that the $\Delta^2$ HWT coefficients in the vector $\vec{\mathbf{w}}$ are on $(\log_2\Delta + 1)^2$ different scale pairs, namely $(m_1, m_2)$ for $m_1, m_2 = -1, 0, 1,\cdots, \log_2\Delta-1$.
Our scheme uses $(\log_2\Delta + 1)^2$ independent
$k$-wise independent hash functions that we denote as $h_{m_1, m_2}(\cdot)$, for $m_1, m_2 = -1, 0, 1, \cdots, \log_2\Delta-1$.    
During the initialization phase, we uniformly randomly seed these $(\log_2\Delta + 1)^2$ hash functions;  once seeded, they are fixed thereafter as usual.  
As mentioned earlier, these seeds correspond to the outcome $\omega$ that fixes (mathematically defines) the HWT coefficient vector $\vec{\mathbf{w}}$.

Our scheme can be stated {\it literally in one sentence:}  Each such (seeded and fixed) $h_{m_1, m_2}(\cdot)$ is solely responsible for hash-generating 
any HWT coefficient on scale $(m_1, m_2)$ that is \emph{participating} (as defined earlier) in answering a range-sum query.  
In other words, for any scale $m_1, m_2 = -1, 0, 1,\cdots, \log_2\Delta-1$, and location $j_1 = 0, 1, \cdots, 2^{m_1^+}-1$, $j_2 = 0, 1, \cdots, 2^{m_2^+}-1$, the value of the HWT coefficient $W^{m_1, m_2}_{j_1, j_2}$ is {\it mathematically defined} as $g(h_{m_1, m_2}(j_1, j_2))$, where $g(\cdot)$ is the aforementioned 
deterministic function (that maps an integer to a Gaussian RV).  Hence our scheme has a much lower space complexity of $O(\log^2 \Delta)$, 
for remembering the seeds of the $O(\log^2 \Delta)$ hash functions.

The following theorem states that our scheme achieves its intended objective of ensuring that the $\Delta^2$ underlying RVs {\it mathematically defined} by it are $k$-wise independent.
In this theorem and proof, we denote the vector of $\Delta^2$ HWT coefficients and the vector of $\Delta^2$ underlying RVs both mathematically defined by our scheme 
as $\vec{\mathbf{v}}$ and $\vec{\mathbf{z}}$, respectively.   
We do so to distinguish this vector pair from the original vector pair $\vec{\mathbf{w}}$ and $\vec{\mathbf{x}}$ that are mathematically defined by the idealized scheme 
(that guarantees mutual independence). 
Recall that $\vec{\mathbf{z}} = M^T\vec{\mathbf{v}}$ and $\vec{\mathbf{x}} = M^T\vec{\mathbf{w}}$, where $M = H^{\otimes 2}$ is the $2$D HWT matrix, and that 
$\vec{\mathbf{x}}$ is comprised of i.i.d. standard Gaussian RVs.

\begin{theorem}\label{th:kwise}
	The vector $\vec{\mathbf{z}}$ is comprised of $k$-wise independent standard Gaussian RVs.
\end{theorem}

\begin{proof}
It suffices to prove that \emph{any} $k$ distinct scalars in $\vec{\mathbf{z}}$ -- say the $(i_1)^{th}$, $(i_2)^{th}$, $\cdots$, $(i_k)^{th}$ scalars -- are i.i.d. standard Gaussian. 
Let $\vec{\mathbf{z}}'$ be the $k$-dimensional vector comprised of these $k$ scalars.  
Let $(M^T)'$ be the $k\times \Delta^2$ matrix formed by the $(i_1)^{th}, (i_2)^{th}, \cdots, (i_k)^{th}$ rows in $M^T$. 
Then, we have $\vec{\mathbf{z}}' = (M^T)' \vec{\mathbf{v}}$. 
Now let the random vector $\vec{\mathbf{x}}'$ be defined as $(M^T)' \vec{\mathbf{w}}$.  Then $\vec{\mathbf{x}}'$ is comprised of 
$k$ i.i.d. standard Gaussian RVs, as its scalars are a subset of those of $\vec{\mathbf{x}}$.  Hence, to prove that the scalars in $\vec{\mathbf{z}}'$ are i.i.d. standard Gaussian RVs, 
it suffices to prove the claim that $\vec{\mathbf{z}}'$ has the same distribution as $\vec{\mathbf{x}}'$.
	
We prove this claim using~\Cref{ob:indepsum}.  To this end, we 
first write $\vec{\mathbf{z}}'$ and $\vec{\mathbf{x}}'$ each as the sum of $N=(\log_2\Delta+1)^2$ independent random vectors.
Recall that in~\autoref{sec:2dwavelet}, we have classified the HWT coefficients in $\vec{\mathbf{w}}$ and $\vec{\mathbf{v}}$, and the columns of $M^T$ (called HWT vectors there)
into $N$ different $(m_1, m_2)$ scales (colors in~\autoref{fig:colors}).  Recall that $n_{m_1, m_2}$ scalars in $\vec{\mathbf{w}}$ and $\vec{\mathbf{v}}$, and accordingly $n_{m_1, m_2}$ columns of $M^T$,
have scale $(m_1, m_2)$.  
Let $\vec{\mathbf{w}}_{m_1, m_2}$ and $\vec{\mathbf{v}}_{m_1, m_2}$ be the $n_{m_1, m_2}$-dimensional vectors comprised of the coefficients classified to scale $(m_1, m_2)$ in $\vec{\mathbf{w}}$ and 
$\vec{\mathbf{v}}$, respectively.  Let $(M^T)'_{m_1, m_2}$ be the $k\times n_{m_1, m_2}$ matrix comprised of the columns of $(M^T)'$ classified to scale $(m_1, m_2)$.  
Then, 
we have $\vec{\mathbf{z}}' = \sum_{(m_1, m_2)}(M^T)'_{m_1, m_2}\vec{\mathbf{v}}_{m_1, m_2}$ and $\vec{\mathbf{x}}' = \sum_{(m_1, m_2)}(M^T)'_{m_1, m_2}\vec{\mathbf{w}}_{m_1, m_2}$, 
where both summations are over all $N$ scales.  
The $N$ summands in the RHS of the first equation are independent random vectors, because 
for each scale $(m_1, m_2)\in [-1, \log_2\Delta)^2$, all scalars in $\vec{\mathbf{v}}_{m_1, m_2}$ are generated by the same per-scale hash function $h_{m_1, m_2}(\cdot)$, 
which is independent of all $N - 1$ other per-scale hash functions.  The same can be said about the $N$ summands in the RHS of the second equation, 
since $\vec{\mathbf{w}}$ is comprised of i.i.d. RVs by design.  To prove this claim using~\Cref{ob:indepsum}, it remains to prove the fact that for each scale $(m_1, m_2)\in [-1, \log_2\Delta)^2$, the pair of random vectors $(M^T)'_{m_1, m_2}\vec{\mathbf{v}}_{m_1, m_2}$ and $(M^T)'_{m_1, m_2}\vec{\mathbf{w}}_{m_1, m_2}$
	have the same distribution.

This fact can be proved as follows.
Note that for each scale $(m_1, m_2)\in [-1, \log_2\Delta)^2$, 
each row in $(M^T)'_{m_1, m_2}$ has exactly one nonzero scalar, since the corresponding row in $M^T$, or equivalently the corresponding column in $M$, 
has exactly one nonzero scalar at each scale $(m_1, m_2)$, due to~\Cref{lem:eachcolor}.
Therefore, although the number of columns in $(M^T)'_{m_1, m_2}$ can be as many as $O(\Delta^2)$, at most $k$ of them (one for each row), say the  $(\alpha_1)^{th}, (\alpha_2)^{th}, \cdots, (\alpha_k)^{th}$ columns, contain nonzero scalars. 
Then, $(M^T)'_{m_1, m_2}\vec{\mathbf{v}}_{m_1, m_2}$  is a function of only the $(\alpha_1)^{th}, (\alpha_2)^{th}, \cdots, (\alpha_k)^{th}$ scalars in $\vec{\mathbf{v}}_{m_1, m_2}$, and these $k$ scalars are i.i.d. Gaussian RVs since they are all generated by the same $k$-wise independent hash function $h_{m_1, m_2}(\cdot)$. 
Note that $(M^T)'_{m_1, m_2}\vec{\mathbf{w}}_{m_1, m_2}$ is the same   function of the $(\alpha_1)^{th}, (\alpha_2)^{th}, \cdots$, $(\alpha_k)^{th}$ scalars in 
$\vec{\mathbf{w}}_{m_1, m_2}$, which are i.i.d. Gaussian RVs by design.  Hence, $(M^T)'_{m_1, m_2}\vec{\mathbf{v}}_{m_1, m_2}$ has the same distribution as $(M^T)'_{m_1, m_2}\vec{\mathbf{w}}_{m_1, m_2}$.
\end{proof}

\begin{proposition}\label{ob:indepsum}
	Suppose random vectors $\vec{\mathbf{x}}$ and $\vec{\mathbf{z}}$ each is the sum of $N$ \emph{independent} random vectors as follows: $\vec{\mathbf{x}} = \vec{\mathbf{x}}_1 + \vec{\mathbf{x}}_2 + \cdots + \vec{\mathbf{x}}_N$ and $\vec{\mathbf{z}} = \vec{\mathbf{z}}_1 + \vec{\mathbf{z}}_2 + \cdots + \vec{\mathbf{z}}_N$. Then, $\vec{\mathbf{x}}$ and $\vec{\mathbf{z}}$ have the same distribution if each pair of components $\vec{\mathbf{x}}_i$ and $\vec{\mathbf{z}}_i$ have the same distribution, for $i=1,2,\cdots, N$.
\end{proposition}

\begin{lemma}\label{lem:eachcolor}
	Any column of $M = H^{\otimes 2}$, which is equal to $H^{\otimes 2}\mathbbm{1}_{\{\vec{\mathbf{i}}\}}$ for some $\vec{\mathbf{i}} = (i_1, i_2)^T$, has exactly one nonzero scalar on each 2D scale $(m_1, m_2)$. 
\end{lemma}
\begin{proof}
	The 2D indicator vector can be decomposed to the Kronecker product of two 1D indicator vectors as $\mathbbm{1}_{\{\vec{\mathbf{i}}\}}= \mathbbm{1}_{\{i_1\}}\otimes \mathbbm{1}_{\{i_2\}}$, so $H^{\otimes 2}\mathbbm{1}_{\{\vec{\mathbf{i}}\}} = (H\mathbbm{1}_{\{i_1\}}) \otimes (H\mathbbm{1}_{\{i_2\}})$ by~\Cref{th:mixedproduct}.  The claim above follows from~\Cref{lem:onehot}, which implies that $H\mathbbm{1}_{\{i_1\}}$ and $H\mathbbm{1}_{\{i_2\}}$ each
	has exactly one nonzero scalar on each 1D scale.  
	
\end{proof}

\section{Multidimensional Dyadic Simulation}\label{sec:infeasibility}



As explained in~\autoref{sec:dst1D}, in 
one dimension (1D), any dyadic range-sum $S[l, u)$, {\it no matter what} the target distribution is, 
can be computed by performing $O(\log\Delta)$ \emph{binary splits} along the path from the root $S[0, \Delta)$ to the node $S[l, u)$ along the 
dyadic simulation tree (DST). 
Since we have just {\it computationally efficiently} generalized the Gaussian-DST approach (equivalent to the HWT-based approach in the 1D Gaussian case) to any dimension $d \ge 2$,
we wonder whether we can do the same for all target distributions.
By ``computationally efficiently'', we mean that a generalized solution should be able to compute any $d$D range-sum in $O(\log^d\Delta)$ time like in the Gaussian case. 




Unfortunately, it appears hard, if not impossible, to generalize the DST approach to $d$D for arbitrary target distributions.  
We have identified a sufficient condition on the target distribution for such an efficient generalization to exist.
We prove the sufficiency by proposing a DST-based universal algorithmic framework (described in~\autoref{sec:2ddyasim} in the interest of space) that solves the $d$D-ERS problem for 
any target distribution satisfying this condition. 
Unfortunately, so far only two distributions, namely Gaussian and Poisson, are known to satisfy this condition, as is elaborated in~\autoref{sec:casepositive}.
We also describe in~\autoref{sec:casenegative} two example distributions that do not satisfy this sufficient condition, namely Cauchy and Rademacher.
In the following, we specify this condition and explain why it is ``almost necessary''.

For ease of presentation, in the following, we fix the number of dimensions $d$ at 2.
We assume all underlying RVs, $X_{i_1, i_2}$ for $(i_1, i_2)$ in the 2D universe $[0, \Delta)^2$, are i.i.d. with a certain target distribution $X$. 
This assumption is appropriate for our reasoning below about the time complexity of a 2D ERS solution, since
as shown earlier this time complexity is not affected by the strength of the independence guarantee provided, in the cases of Gaussian and Poisson.
In a 2D universe, 
any 2D range can be considered the Cartesian product of its horizontal and vertical 1D ranges.
We say a 2D range is \emph{dyadic} if and only if its horizontal and vertical 1D ranges are both dyadic.
Since any general  (not necessarily dyadic) 1D range can be ``pieced together'' using  $O(\log\Delta)$ 1D dyadic ranges~\cite{rusu-jour}, it is not hard to show, using the Cartesian product argument, that
any general 2D range can be ``pieced together'' using $O(\log^2\Delta)$ 2D dyadic ranges.  
Hence in the following, we focus on the generation of only 2D dyadic range-sums.
We assume all underlying RVs, $X_{i_1, i_2}$ for $(i_1, i_2)$ in the 2D universe $[0, \Delta)^2$, are i.i.d. with a certain target distribution $X$.


We need to introduce some additional notations.  
We define each horizontal strip-sum  $S^H_i \triangleq X_{i, 0} + X_{i, 1} + \cdots + X_{i, \Delta-1}$ for $i\in[0,\Delta)$ as the sum of range $[i, i+1)\times [0, \Delta)$,
and each vertical strip-sum $S^V_i \triangleq X_{0, i} + X_{1, i} + \cdots + X_{\Delta-1,i}$ for $i\in[0,\Delta)$ as the sum of range $[0, \Delta)\times [i, i+1)$. 
We denote as $S$ the total sum of all underlying RVs in the universe, \ie
$S\triangleq \sum_{i_1=0}^{\Delta-1}\sum_{i_2=0}^{\Delta-1}X_{i_1, i_2} = \sum_{i=0}^{\Delta-1} S^H_i = \sum_{i=0}^{\Delta-1} S^V_i$.

Now we are ready to state this sufficient condition.  For ease of presentation, we break it down into two parts.  
The first part, stated in the following formula, states that the vector of vertical strip-sums and the vector of horizontal strip-sums in $[0, \Delta)^2$ 
are conditionally independent given the total sum $S$.  
\begin{equation}\label{condindep}
(S^V_0, S^V_1, \cdots, S^V_{\Delta-1})  \indep (S^H_0, S^H_1, \cdots, S^H_{\Delta-1}) \mid  S. 
\end{equation}
The second part is that this conditional independence relation holds for the two corresponding vectors in any 2D dyadic range (that is not necessarily a square).  
Intuitively, this condition says that how a 2D dyadic range-sum is split horizontally is conditionally independent (upon this 2D range-sum) of how it is split vertically. 
Roughly speaking, this condition implies that the 1D-DST governing the horizontal splits is conditionally independent of the other 1D-DST governing the vertical splits.  
Hence, our DST-based universal algorithmic framework for 2D can be viewed as the Cartesian product of the two 1D-DSTs,  
as will be elaborated in~\autoref{sec:2ddyasim}.

In the following, we offer some intuitive evidence why this condition is likely necessary.  
Without loss of generality, we consider the generation of an arbitrary horizontal strip-sum $S^H_{i_1}$ conditional on the vector of vertical strip-sums $(S^V_0, S^V_1, \cdots, S^V_{\Delta-1})$. 
Suppose~(\ref{condindep}) does not hold, which means $(S^H_0, S^H_1, \cdots, S^H_{\Delta-1})$ is not conditionally independent of $(S^V_0, S^V_1, \cdots, S^V_{\Delta-1})$ given $S$.  
Then the distribution of $S^H_{i_1}$ is arguably parameterized by the values (realizations) of all $\Delta$ vertical strip-sums $S^V_0, S^V_1, \cdots, S^V_{\Delta-1}$, since $S^H_{i_1}$ and any vertical strip-sum $S^V_{i_2}$ for $i_2\in[0,\Delta)$ are in general dependent RVs by~\Cref{th:hvdep} (See~\autoref{appd:l17} for its nontrivial proof).  
Hence, unless some magic happens (which we cannot rule out rigorously), to generate (realize) the RV $S^H_{i_1}$, conceivably we need to first realize all $\Delta$ RVs $(S^V_0, S^V_1, \cdots, S^V_{\Delta-1})$, the time complexity of which is $\Omega(\Delta)$.  
 
 \begin{theorem}\label{th:hvdep}
 	For any $(i_1, i_2)$ in $[0, \Delta)^2$, $S^H_{i_1}$ and $S^V_{i_2}$ are dependent RVs unless the target distribution $X$ satisfies $\Pr[X=c]=1$ for some constant $c$.
 \end{theorem}

 
 



\section{Conclusion}\label{sec:conclusion}
In this work, we propose novel solutions to $d$D-ERS for RVs that have Gaussian or Poisson distribution. 
Our solutions are the first ones that compute any multi-dimensional range-sum in polylogarithmic time. 
Our $d$D Gaussian-ERS scheme solves the long-standing open problem of efficiently answering approximate range-sum queries over a multidimensional data cube.
We develop a novel $k$-wise independence theory that provides both high computational efficiencies and strong provable independence guarantees. 
Finally, 
 we show that when the underlying distribution satisfies a sufficient and likely necessary condition, its DST-based 1D-ERS solution can be generalized to higher dimensions.

\bibliographystyle{plainurl}
\bibliography{bibs/draft}
\appendix

\section{Two Known Positive Cases:  Gaussian and Poisson}\label{sec:casepositive}

When the target distribution is \emph{Gaussian}, 
it can be shown that the vector of vertical strips $(S^V_0, S^V_1, \cdots, S^V_{\Delta-1})$ is a function only of all the HWT coefficients in the form of $W^{-1, m_2}_{0, j_2}$ (for some scale $m_2$ and location $j_2$), which are the $\Delta$ scalars on the first row (from the top) of $\vec{\mathbf{w}}$  arranged into a $\Delta \times \Delta$ matrix 
(like that shown in~\autoref{fig:colors});  we denote these HWT coefficients  as a vector $\vec{\mathbf{w}}^V$. 
It can be shown that the vector of horizontal strips $(S^H_0, S^H_1, \cdots, S^H_{\Delta-1})$ is a function only of all the HWT coefficients in the form of $W^{m_1, -1}_{j_1, 0}$ (for some scale $m_1$ and location $j_1$), which are the $\Delta$ scalars on the first column (from the left) of $\vec{\mathbf{w}}$ in the same matrix form as above;
we denote these HWT coefficients  as a vector $\vec{\mathbf{w}}^H$. 
Since all the $\Delta^2$ HWT coefficients are i.i.d. Gaussian when the $\Delta^2$ underlying RVs are i.i.d. Gaussian (due to~\Cref{lem:correctness}), 
$(S^V_0, S^V_1, \cdots, S^V_{\Delta-1})$ and $(S^H_0, S^H_1, \cdots, S^H_{\Delta-1})$ are independent conditioned upon $S$, because 
the two vectors $\vec{\mathbf{w}}^H$ and $\vec{\mathbf{w}}^V$ share only a single common element:  $W^{-1,-1}_{0,0} = 1/\Delta \cdot  \sum_{i_1=0}^{\Delta-1}\sum_{i_2=0}^{\Delta-1}X_{i_1, i_2} = S/\Delta$.


To prove that \emph{Poisson} satisfies the sufficient condition, we need to introduce the following ``balls-into-bins'' process.  Mathematically but not computationally, 
we independently throw $S$ balls each uniformly and randomly into one of the $\Delta^2$ bins arranged as a $\Delta\times \Delta$ matrix indexed by $(i_1, i_2)\in [0, \Delta)^2$.
Each underlying $X_{i_1, i_2}$ is defined as the number of balls that end up in the bin indexed by $(i_1, i_2)$.
It is not hard to verify that, given the total sum $S$, 
the $\Delta^2$ underlying RVs generated by the ``balls-into-bins'' process have the same conditional (upon their total sum being $S$) joint distribution as
$\Delta^2$ i.i.d. Poisson RVs.  
Note that throwing a ball uniformly and randomly into a bin consists of the following two independent steps.
The first step is to select $i_1$ uniformly from $[0, \Delta)$ and throw the ball to the $(i_1)^{th}$ row (thus increasing the horizontal strip-sum $S^H_{i_1}$ by 1). 
The second step is to select $i_2$ uniformly from $[0, \Delta)$ and throw the ball to the $(i_2)^{th}$ column (thus increasing the vertical strip-sum $S^V_{i_2}$ by 1). 
Therefore, $(S^V_0, S^V_1, \cdots, S^V_{\Delta-1})$ and $(S^H_0, S^H_1, \cdots, S^H_{\Delta-1})$ are conditionally independent given $S$, because these two 
steps, in throwing each of the $S$ balls, are independent.

\section{Two Example Negative Cases:  Rademacher and Cauchy}\label{sec:casenegative}



As the negative cases are shown by counterexamples, we set $\Delta$ to a small number $2$ to make this job easier.
In this case we are dealing with only 4 underlying RVs $X_{0,0}$, $X_{0,1}$, $X_{1,0}$, and $X_{1,1}$.  
We start with the target distribution being \emph{Rademacher} ($\Pr[X = 1] = \Pr[X = 0] = 0.5$).  We prove by contradiction.  Suppose $S^V_0$ is independent of $(S^H_0, S^H_1)$ conditioned on $S$.
We have $\Pr(S^V_0 = 0 \mid S=0) \Pr(S^H_0 = 2, S^H_1 = -2 \mid S=0) = \Pr(S^V_0 = 0 \mid S^H_0 = 2, S^H_1 = -2, S=0)\Pr(S^H_0 = 2, S^H_1 = -2 \mid S=0)$, 
since both the LHS and the RHS are equal to $\Pr(S^V_0 = 0, S^H_0 = 2, S^H_1 = -2 \mid S=0)$.  Since the LHS and the RHS contain a common factor $\Pr(S^H_0 = 2, S^H_1 = -2 \mid S=0)\neq 0$,  
we obtain $\Pr(S^V_0 = 0 \mid S=0) = \Pr(S^V_0 = 0 \mid S^H_0 = 2, S^H_1 = -2, S=0)$ by removing the common factor.
It is not hard to calculate using~(\ref{eq:dstsplit}) that $\Pr(S^V_0 = 0 \mid S=0) = 2/3$.  However, $\Pr(S^V_0 = 0 \mid S^H_0 = 2, S^H_1 = -2, S=0) = 1$, because 
$S^H_0 = 2$ implies $X_{0, 0} = 1$ and $S^H_1 = -2$ implies $X_{1, 0} = -1$, and as a result $S^V_0 = X_{0, 0} + X_{1, 0} = 0$ with probability 1.  
Therefore, we have a contradiction.

\begin{remark}
	The authors in~\cite{muthukrishnanhistogram} were looking for a 2D-ERS solution for Rademacher RVs with (approximately) $4$-wise independence guarantee.  The fact that~(\ref{condindep}) does not hold likely rules out a DST-based Rademacher-ERS
	solution that works in 2D.  However, it does not rule out a generalization of 1D ECC-based Rademacher-ERS schemes such as EH3~\cite{feigenbaum} to 2D, although no such generalization has been discovered to this day~\cite{rusu-jour}.
\end{remark}

Now we move on to the target distribution being \emph{Cauchy}.  
Suppose $S^H_0=z$ and $S^H_1 = -z$ for a real number $z$ such that $S = S^H_0 + S^H_1 = 0$. It can be shown that (\emph{e.g.}, in~\cite{1ddyadic}) 
$X_{0,0}$ has conditional pdf $f(x\mid z) = (4+z^2) /(2\pi(1+x^2)(1+(z-x)^2))$, and that $X_{1,0}$ has pdf $f(x\mid-z)$. Hence, $S^V_0$ has pdf $g(x\mid z) = f(x\mid z) \ast f(x\mid -z)$, where $\ast$ refers to the convolution on $x$. 
We consider the value of $g(x\mid z)$ at $x=0$. 
We have $g(0\mid z) = \int_{-\infty}^{\infty}  f(x\mid z)\cdot f(-x\mid -z)\mathrm{d}x = (4+z^2)^2/(4\pi^2)\cdot \int_{-\infty}^{\infty} 1/((1+x^2)^2(1+(z-x)^2)^2)\cdot\mathrm{d}x = (20+z^2)/(4\pi(4+z^2))$. 
Since $g(0\mid z)$ is still a function of $z$, we conclude $S^V_0$ is dependent on $(S^H_0 , S^H_1)$ given $S$.

\section{Universal Algorithmic Framework}\label{sec:2ddyasim}


In this section, we describe the universal algorithmic framework, for all target distributions that satisfy the aforementioned sufficient condition (which so far include only Gaussian and Poisson), 
that can be instantiated to compute the sum of any 2D dyadic range in $O(\log^2\Delta)$ time;  its space complexities for providing the two aforementioned types of 
independence guarantees (mutual independence and $k$-wise independence) are the same as those of the 2D ERS-Gaussian solution.
Interestingly, the time complexity of computing any general 2D range-sum remains $O(\log^2\Delta)$, as we will explain shortly.
This framework is simply to perform a 4-way-split procedure recursively on all ancestors (defined next) of the 2D dyadic range whose sum is to be generated.
In a 4-way-split procedure, we compute a 2D dyadic range-sum say $S[\vec{\mathbf{l}}, \vec{\mathbf{u}})$ by splitting
the sum of its {\it lattice grandparent}, 
which is the unique (if it exists) 2D dyadic range that contains $[\vec{\mathbf{l}}, \vec{\mathbf{u}})$ and is twice as large along both horizontal and vertical 
dimensions, into four pieces.  

\begin{figure}
	\includegraphics[width=0.45\textwidth]{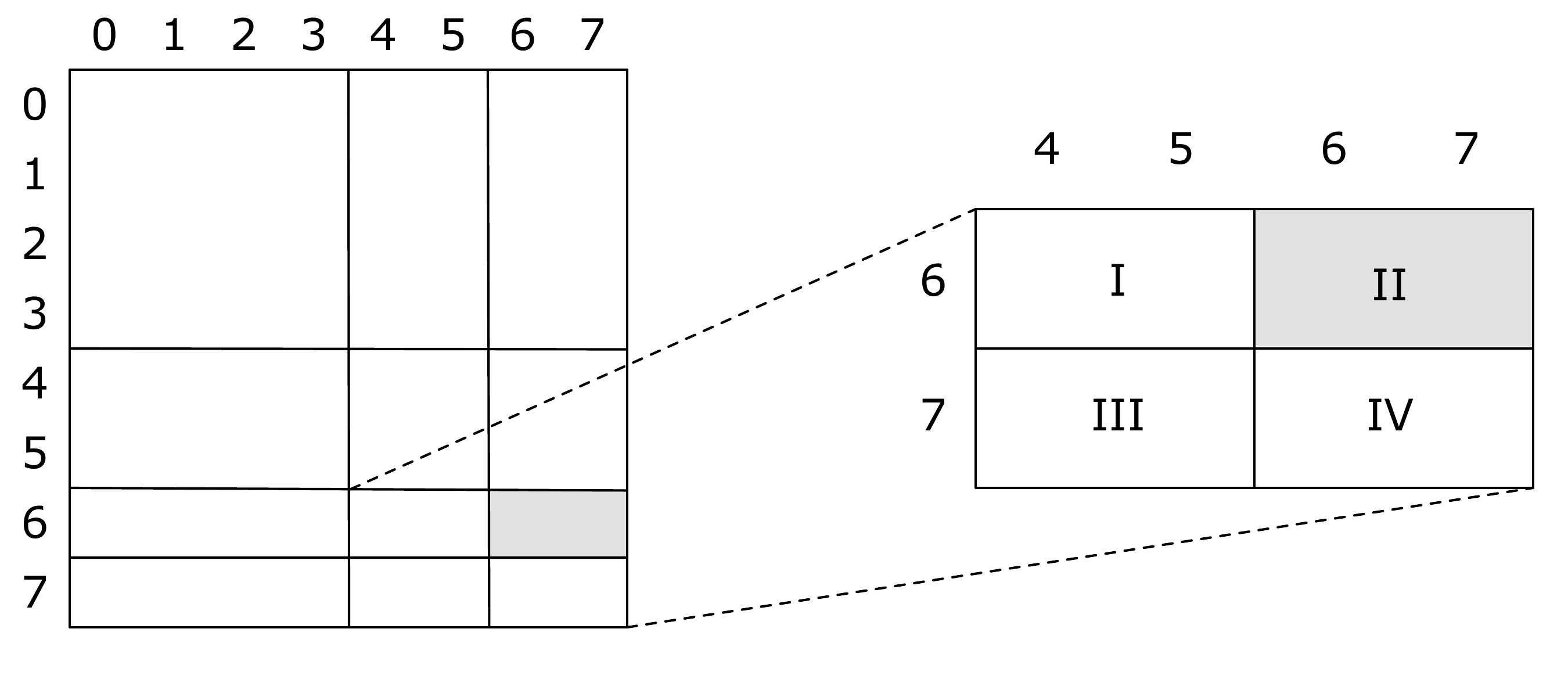}
	\caption{An example of the universal algorithmic framework in 2D with $\Delta=8$.} \label{fig:4parents}
\end{figure}

We illustrate this procedure by a tiny example (with universe size $\Delta=8$) 
shown in~\autoref{fig:4parents}.  Our computation task is to generate the sum of the 2D dyadic range $[6,7)\times [6,8)$ that is shadowed and marked as region II in~\autoref{fig:4parents}.  
We call the region I + II
its \emph{horizontal parent}, since it contains II and is twice as large in the horizontal dimension.
Similarly, we call the region II + IV its \emph{vertical parent}.  Its lattice grandparent is the union of the 4 regions I + II + III + IV.  
We refer to the horizontal and the vertical parents, and the lattice grandparent,
of a 2D dyadic range as its \emph{direct generation ancestors} (DGAs).  

In this example, 
the sum of the 2D range $[6,7)\times [6,8)$ is to be generated according to a conditional distribution 
parameterized by the sums of its three DGAs.  To do so, however, the sum of each such DGA is
to be generated according to the sums of the DGA's three DGAs, and so on.  
Given any 2D dyadic range, since it has $O(\log^2\Delta)$ distinct {\it ancestors} (DGAs, DGAs' DGAs, and so on), 
we can generate its sum in $O(\log^2\Delta)$ time, by arranging the computations of the sums of these $O(\log^2\Delta)$ ancestors in the 
\emph{dynamic programming} order.  
One can think of these $O(\log^2\Delta)$ 4-way-split procedures as 
the Cartesian product of the $O(\log\Delta)$ binary splits along the horizontal 1D-DST
and those along the vertical 1D-DST.  As explained in~\autoref{sec:infeasibility}, the sufficient condition makes taking this Cartesian product
possible.



In a 4-way-split procedure, degenerated cases arise when the 2D dyadic range whose sum is to be generated 
spans the entire 1D universe on either dimension or on both dimensions.  
In the former case, 
this 2D dyadic range has only one parent and no lattice grandparent.   
For example, the range $[4,6)\times[0,8)$ has only a vertical parent $[4,8)\times[0,8)$.  In this case, the 4-way-split degenerates to the 2-way split in 1D as already
explained in~(\ref{eq:dstsplit}).  In the latter case, which is the only boundary condition of the aforementioned dynamic programming,
the 2D dyadic range is the entire 2D universe.  In this case, its
sum is directly generated from the distribution $X^{*\Delta^2}$.
In the cases of both Gaussian and Poisson target distributions, $X^{*\Delta^2}$ takes a simple form and can be generated efficiently~\cite{boxmuller, marsaglia}.




So far we have only explained why a 2D dyadic range-sum can be generated in $O(\log^2\Delta)$ time.
In fact, any general 2D range-sum can also be generated in $O(\log^2\Delta)$ time due to the following two ``2D facts''.
First, any general 2D range can be ``pieced together'' using $O(\log^2\Delta)$ 2D dyadic ranges, as explained in~\autoref{sec:infeasibility}.
Second, these $O(\log^2\Delta)$ 2D dyadic ranges together have only $O(\log^2\Delta)$ distinct (2D DST) ancestors for the following reason.
The following ``1D fact'' was proved in Corollary 1 in~\cite{1ddyadic}:  Every 1D general range can be ``pieced together'' using $O(\log\Delta)$ 1D dyadic ranges and 
these $O(\log\Delta)$ 1D dyadic ranges share $O(\log\Delta)$ common ancestors on the 1D DST.  
Since a 2D general range is by definition the Cartesian product of two 1D ranges (namely vertical and horizontal), ``multiplying'' this 1D fact for the vertical 1D range by the 1D fact for the horizontal 1D range
proves the second 2D fact.   

We can generalize the universal algorithmic framework  above to $d$D.  In the generalized framework, the 2-way conditional independence relation in 2D, namely formula~\autoref{condindep}, 
becomes a $d$-way conditional  independence relation 
in $d$D, and the $4$-way-split procedure in 2D becomes a $2^d$-way-split procedure in $d$D.  We omit the details of this generalization in the interest of space.

\noindent
\textbf{4-Way-Split Procedure for Poisson}: While the algorithmic framework of performing 4-way splits recursively is the same for any target distribution that satisfies the sufficient condition, 
the exact 4-way-split procedure is different for each such target distribution.  In the following, we specify only the
4-way-split procedure for Poisson, since that for Gaussian has already been seamlessly ``embedded'' in the HWT-based solution.  
Suppose given a 2D dyadic range $[\vec{\mathbf{l}}, \vec{\mathbf{u}})$, we need to compute $S[\vec{\mathbf{l}}, \vec{\mathbf{u}})$
by performing a 4-way split.  
We denote as $S_h$, $S_v$, and $S_g$ the sums of the horizontal parent, 
the vertical parent, and the lattice grandparent of $[\vec{\mathbf{l}}, \vec{\mathbf{u}})$, respectively.  
Then $S[\vec{\mathbf{l}}, \vec{\mathbf{u}})$ is generated using the following conditional distribution:  
\begin{equation}
	\Pr\left(S[\vec{\mathbf{l}}, \vec{\mathbf{u}}) = x\mid S_h=s_h, S_v = s_v, S_g = s_g\right)  = 1/c \cdot f(x\mid s_h, s_v, s_g), 
\end{equation}
where $f(x\mid s_h, s_v, s_g) =1/\left(x!(s_h-x)!(s_v-x)!(s_g-s_h-s_v+x)!\right)$ and $c\triangleq \sum_{y=0}^\infty f(y\mid s_h, s_v, s_g)$ is a normalizing constant. 
As explained, in a degenerated case, this split is the same as a 2-way split in the 1D case.  
In the case of Poisson, by plugging the Poisson pmfs $\phi_n(x) = e^{-n}n^x/(x!)$ and $\phi_{2n}(x) = e^{-2n}(2n)^x/(x!)$ into~(\ref{eq:dstsplit}), we obtain
$f(x\mid z) = 2^{-z}\binom{z}{x}$, which is the pmf of $Binomial(z, 1/2)$.

\section{Proof of Theorem~\ref{th:hvdep}}\label{appd:l17}
\begin{proof}
	Note that for any $(i_1, i_2)$ in $[0, \Delta)^2$, $S^H_{i_1}$ intersects $S^V_{i_2}$ on exactly one underlying RV $X_{i_1, i_2}$, so $S^H_{i_1} -X_{i_1, i_2}, S^V_{i_2}-X_{i_1, i_2}$, and $X_{i_1, i_2}$ are independent. By~\Cref{lem:xyz}, $S^H_{i_1}$ and $S^V_{i_2}$ are always dependent unless there exists some constant $c$ such that  $\Pr[X_{i_1, i_2}=c]=1$. 
\end{proof}

\begin{lemma}\label{lem:xyz}
	Let $X, Y, Z$ be three  independent RVs.  Suppose $X+Z$ and $Y+Z$ are independent.  Then there exists some constant $c$ such that $\Pr[Z=c]=1$.
\end{lemma}
\begin{proof}
	Let $a, b$ be two arbitrary real numbers. Denote as $F_X(\cdot)$ and $F_Y(\cdot)$  the cumulative distribution functions (cdfs) of $X$ and $Y$, respectively. Denote as $\mathbbm{1}_{\mathcal{E}}$ the indicator RV  of an event $\mathcal{E}$, which has value $1$ if  $\mathcal{E}$ happens and $0$ otherwise. Then, $\Pr(X+Z\leq a) = E[\mathbbm{1}_{X+Z\leq a}] = E[E[\mathbbm{1}_{X\leq a-Z}|Z]] = E[F_X(a-Z)]$, where the second equation is by the total expectation formula, and the third equation is because $X$ and $Z$ are independent (so the conditional cdf of $X|Z$ is also $F_X(\cdot)$). Similarly, we have $\Pr(Y+Z\leq b) = E[F_Y(b-Z)]$ and $\Pr(X+Z\leq a, Y+Z\leq b) = E[F_X(a-Z)F_Y(b-Z)]$. Therefore, $\Pr(X+Z\leq a, Y+Z\leq b) - \Pr(X+Z\leq a)\Pr(Y+Z\leq b) = E[F_X(a-Z)F_Y(b-Z)] - E[F_X(a-Z)] E[F_Y(b-Z)] = Cov(F_X(a-Z), F_Y(b-Z))$. This covariance always exists since cdfs are bounded functions. Since $X+Z$ and $Y+Z$ are independent, $\Pr(X+Z\leq a, Y+Z\leq b)	= \Pr(X+Z\leq a)\Pr(Y+Z\leq b)$, so $Cov(F_X(a-Z), F_Y(b-Z)) = 0$.

	We prove by contradiction. Suppose $Z$ is not with probability $1$ a constant. Then for any set $\mathcal{I}$ such that $\Pr(Z\in \mathcal{I})= 1$, $\mathcal{I}$ must contain at least two numbers. 
	Let $z_1 < z_2$ be two such numbers such that the pdf (or pmf) of $Z$ is nonzero at  $z_1$ and $z_2$.	
	Since both $F_X(a-Z)$ and $F_Y(b-Z)$ are non-increasing functions of $Z$, and $Cov(F_X(a-Z), F_Y(b-Z)) = 0$, either $F_X(a-Z)$ or $F_Y(b-Z)$ 
	must be a constant with probability $1$ by~\Cref{th:chebyint}. Without loss of generality, suppose $F_X(a-Z)$ is a constant with probability $1$.  
	Then, $F_X(a-z_1) = F_X(a-z_2)$ holds for arbitrary values of $a$. 
	Define a sequence $x_0 \triangleq z_1, x_1 \triangleq z_1 + (z_2-z_1), x_2 \triangleq z_1 + 2\cdot(z_2 - z_1), 
	\cdots, x_k \triangleq  z_1 + k\cdot(z_2 - z_1), \cdots$ such that we have $\lim_{i\to\infty} x_i = \infty$. Then, we have $F_X(x_0) = F_X(x_1) = F_X(x_2) = \cdots$, where the $k$-th equation $F_X(x_{k-1}) = F_X(x_k)$ is by applying $F_X(a-z_2) = F_X(a-z_1)$ with $a = 2z_1 + k\cdot(z_2 - z_1)$.  Since $\lim_{i\to\infty} x_i = \infty$, we have $ F_X(x_0) = \lim\inf_{i\to\infty}F_X(x_i) \geq \lim\inf_{x\to\infty}F_X(x)$. Similarly, we can prove $F_X(x_0) \leq\lim\sup_{x\to-\infty}F_X(x)$. By the properties of cdf, we have $\lim\inf_{x\to\infty}F_X(x) = 1$ and $\lim\sup_{x\to-\infty}F_X(x)=0$. As a result, we have proved $0 \geq 1$, which is a contradiction.
	
\end{proof}

\begin{theorem}[\cite{chebyshevint}]\label{th:chebyint}
	Let $Z$ be an RV, and $f(\cdot)$ and $g(\cdot)$ be two non-increasing functions.  Then the covariance $Cov(f(Z), g(Z)) \geq 0$ if it exists. Furthermore, $Cov(f(Z), g(Z)) = 0$ if and only if either $f(Z)$ or $g(Z)$ is a constant with probability $1$.
\end{theorem}

\end{document}